\documentclass{llncs}

\usepackage{pdflscape}
\usepackage{csvsimple}
\usepackage{longtable}
\usepackage{rotating}
\usepackage{subcaption}                        
\usepackage[]{algorithm2e}
\usepackage{wrapfig}
\usepackage{hyperref}
\usepackage{xfrac}
\usepackage{tensor}
\usepackage{mathtools}
\usepackage{todonotes}
\usepackage{mathtools,amssymb}
\usepackage{nicefrac}
\usepackage{graphicx}
\usepackage{array}
\usepackage{stmaryrd}
\usepackage{xspace}
\usepackage{infer}
\usepackage{galois}
\usepackage{marginnote}
\usepackage{tablefootnote}
\usepackage{xcolor}
\usepackage{url}
\usepackage{csquotes}
\usepackage{adjustbox}
\usepackage{paralist}

\usepackage{tikz}
\usetikzlibrary{arrows,decorations.pathmorphing,positioning,fit,trees,shapes,shadows,automata,calc,shapes.multipart,backgrounds}

\newcommand{\SUMPAGE}{14}

\newcommand{\Nats}{\mathbb{N}}

\newcommand{\iH}{H}

\newcommand{\vNull}{\ensuremath{v_{\Null}}}
\newcommand{\I}{I}
\newcommand{\iVar}{\ensuremath{\nu}}
\newcommand{\iBot}{\ensuremath{z}}

\newcommand{\HC}{\ensuremath{\textnormal{\textbf{HC}}}}
\newcommand{\IHC}{\ensuremath{\textnormal{\textbf{IHC}}}}
\newcommand{\iK}{K}
\newcommand{\iR}{R}

\newcommand{\Proj}[2]{#1 \upharpoonright #2}
\newcommand{\Elem}[2]{#1(#2)}
\newcommand{\Replace}[3]{\ensuremath{#1\left[#2 \mapsto #3\right]}}
\newcommand{\Rule}[2]{#1 \,\to\, #2}
\newcommand{\PS}[1]{\mathcal{P}(#1)}
\newcommand{\PSF}[1]{\mathcal{P}_{\textit{finite}}(#1)}
\newcommand{\FComp}{\,\fatsemi\,}
\newcommand{\FOrder}{\,\dot{\subseteq}\,}
\newcommand{\LOrder}[1]{\,\sqsubseteq_{#1}\,}

\newcommand{\Id}[1]{\textit{id}_{#1}}

\newcommand{\Mod}[2]{\textit{Mod}\llbracket #1 \rrbracket(#2)}

\newcommand{\Fields}{\ensuremath{\textnormal{\textbf{Fields}}}}
\newcommand{\Null}{\ensuremath{\textnormal{\texttt{null}}}}
\newcommand{\Var}{\ensuremath{\textnormal{\textbf{Var}}}}

\newcommand{\Types}{\ensuremath{\textnormal{\textbf{Types}}}}

\newcommand{\Rank}{\ensuremath{\textit{rank}}}
\newcommand{\V}{\ensuremath{V}}
\newcommand{\E}{\ensuremath{E}}
\newcommand{\Att}{\ensuremath{\textit{att}}}
\newcommand{\Lab}{\ensuremath{\textit{lab}}}
\newcommand{\Iso}{\cong}
\newcommand{\Ind}{\ensuremath{\textit{ind}}}
\newcommand{\Ext}{\ensuremath{\textit{ext}}}

\newcommand{\iG}{\ensuremath{G}}
\newcommand{\iDerive}[1]{\Rightarrow_{#1}}
\newcommand{\rDerive}[1]{\,{_{#1}\!\!\Leftarrow}}
\newcommand{\niDerive}[1]{\nRightarrow_{#1}}
\newcommand{\nrDerive}[1]{\,{_{#1}\!\!\nLeftarrow}}

\newcommand{\iL}[1]{L_{#1}}
\newcommand{\rL}[1]{L^{-1}_{#1}}
\newcommand{\Lang}[2]{L_{#1}(#2)}
\newcommand{\iLang}[2]{L^{-1}_{#1}(#2)}

\newcommand{\iC}{\ensuremath{C}}

\newcommand{\iCon}[1]{\ensuremath{\rightrightarrows_{#1}}}
\newcommand{\iAbs}[1]{\ensuremath{\,_{#1}\!\!\leftleftarrows}}
\newcommand{\niCon}[1]{\ensuremath{\not\rightrightarrows_{#1}}}
\newcommand{\niAbs}[1]{\ensuremath{\,_{#1}\!\!\not\leftleftarrows}}
\newcommand{\iGL}[1]{\textit{GL}_{#1}}
\newcommand{\rGL}[1]{\textit{GL}^{-1}_{#1}}

\newcommand{\Mater}[2]{\ensuremath{\textnormal{\textit{materialize}}\llbracket #2 \rrbracket}}
\newcommand{\Canon}[2]{\ensuremath{\textnormal{\textit{canonicalize}}}\llbracket #2 \rrbracket}

\newcommand{\B}{B}

\newcommand{\Pt}{P}

\newcommand{\con}[1]{\gamma}
\newcommand{\abs}[1]{\alpha}
\newcommand{\DomCon}{\textnormal{\textbf{Con}}}
\newcommand{\DomAbs}{\textnormal{\textbf{Abs}}}

\newcommand{\T}{T}

\newcommand{\Z}{s}

\newcommand{\Lfp}{\textit{lfp}}

\newcommand{\Progs}{\ensuremath{\textnormal{\texttt{Progs}}}}
\newcommand{\PExps}{\ensuremath{\textnormal{\texttt{PExp}}}}
\newcommand{\BExps}{\ensuremath{\textnormal{\texttt{BExp}}}}

\newcommand{\pP}{P}
\newcommand{\pE}{\textit{Ptr}}
\newcommand{\pB}{B}
\newcommand{\pF}{f}
\newcommand{\pBot}{\bot}
\newcommand{\True}{\textnormal{\texttt{true}}}
\newcommand{\False}{\textnormal{\texttt{false}}}

\newcommand{\Left}{\texttt{{\color{\lColor}left}}}
\newcommand{\Parent}{\texttt{\color{\pColor}parent}}
\newcommand{\Right}{\texttt{\color{\rColor}right}}
\newcommand{\FL}[1]{\F{#1}{\Left}}
\newcommand{\FP}[1]{\F{#1}{\Parent}}
\newcommand{\FR}[1]{\F{#1}{\Right}}

\newcommand{\Skip}{\textnormal{\texttt{noop}}}
\newcommand{\Ass}[2]{#1\,\textnormal{\texttt{=}}\,#2}
\newcommand{\New}[1]{\textnormal{\texttt{new}}(#1)}
\newcommand{\Seq}[2]{#1\textnormal{\texttt{;}}#2}
\newcommand{\Ite}[3]{\textnormal{\texttt{if}}\,(#1)\, \{ #2 \}\, \textnormal{\texttt{else}}\, \{ #3 \}}
\newcommand{\WHILE}[1]{\textnormal{\texttt{while}}\,(#1)\,}
\newcommand{\WhileDo}[2]{\textnormal{\texttt{while}}\,(#1)\, \{ #2 \}}
\newcommand{\WhileKDo}[3]{\textnormal{\texttt{while}}^{#3}\,(#1)\, \{ #2 \}}

\newcommand{\NEquals}[2]{#1 \,!\!=\, #2}
\newcommand{\Equals}[2]{#1 = #2}
\newcommand{\Con}[2]{#1 \wedge #2}
\newcommand{\Not}[1]{\neg #1}

\newcommand{\F}[2]{#1.#2}

\newcommand{\ESem}[1]{\mathcal{V}\llbracket #1 \rrbracket}
\newcommand{\BSem}[1]{\mathcal{B}\llbracket #1 \rrbracket}
\newcommand{\PSem}[1]{\mathcal{C}\llbracket #1 \rrbracket}
\newcommand{\PASem}[1]{\mathcal{A}\llbracket #1 \rrbracket}

\newcommand{\SemIte}[3]{\textnormal{\text{if}}\,(#1)\,\textnormal{\text{then}}\,#2\, \textnormal{\text{else}}\,#3}
\newcommand{\SetVar}[2]{\textnormal{\text{setVar}}[#1,#2]}
\newcommand{\AddVar}[1]{\textnormal{\text{addNode}}[#1]}
\newcommand{\SetField}[3]{\textnormal{\text{setField}}[#1,#2,#3]}

\tikzstyle{node} = [circle, draw, fill=white, text=black, minimum size = 5mm]
\tikzstyle{pointer} = [->]
\tikzstyle{nonterminal} = [rectangle, draw, fill=black!25, text=black, minimum size = 5mm]
\tikzstyle{variables} = [diamond, draw, fill=white, text=black, minimum size = 5mm]
\tikzstyle{param} = [rectangle, draw, dashed, minimum size = 5mm]

\newcommand{\tikzdefault}[1]{\begin{tikzpicture}[shorten >=1pt, semithick, every text node part/.style={align=center}] #1 \end{tikzpicture}}

\newcommand{\nil}[2]{\node[node] (#1) [#2, fill=black] {};}
\newcommand{\nilext}[2]{\node[node] (#1) [#2, fill=black,text=white] {1};}
\newcommand{\location}[3]{\node[node] (#1) [#2] {#3};}

\newcommand{\lpointer}[6]{\path (#1) edge[pointer, #2] node[draw=none,#6] (#3) {#5} (#4);}
\newcommand{\variable}[3]{\node[variables] (#1) [#2 of = #3] {$#1$}; \path[-] (#1) edge (#3);}

\newcommand{\nonterminal}[4]{\node[nonterminal] (#1) [#2 of = #3] {$#4$};}
\newcommand{\tentacle}[4]{\path (#1) edge[#4] node[above] {$#2$} (#3);}
\newcommand{\tentacleB}[4]{\path (#1) edge[#4] node[below] {$#2$} (#3);}
\newcommand{\tentacleL}[4]{\path (#1) edge[#4] node[left] {$#2$} (#3);}
\newcommand{\tentacleR}[4]{\path (#1) edge[#4] node[right] {$#2$} (#3);}

\tikzstyle{violation}=[draw=black!15,line width=9mm, 
  shorten >= 0pt, anchor=base, cap=round, join=round]

\newcommand{\lColor}{black}
\newcommand{\rColor}{black}
\newcommand{\pColor}{black}

\newcommand{\lEdge}{dashed}
\newcommand{\pEdge}{}
\newcommand{\rEdge}{dotted}

\newcommand{\cEdge}{dash dot}
\newcommand{\cColor}{gray}

\allowdisplaybreaks

\begin{document}

\mainmatter

\title{Graph-Based Shape Analysis Beyond Context-Freeness}
\titlerunning{Heap Abstraction Beyond Context-Freeness}

\author{Hannah Arndt \and Christina Jansen \and Christoph Matheja \and Thomas Noll}
\tocauthor{Hannah Arndt\thanks{DFG}, Christina Jansen, Christoph Matheja\thanks{DFG}, Thomas Noll}
\institute{RWTH Aachen University, Germany}

\maketitle

\begin{abstract}
We develop a shape analysis for reasoning about \emph{relational properties} of data structures. 
Both the concrete and the abstract domain are represented by hypergraphs.
The analysis is parameterized by user-supplied \emph{indexed graph grammars} to guide concretization and abstraction.
This novel extension of context-free graph grammars is powerful enough to model complex data structures such as balanced binary trees with parent pointers, while preserving most desirable properties of context-free graph grammars. 

One strength of our analysis is that no artifacts apart from grammars are required from the user; it thus offers a high degree of automation.
We implemented our analysis and successfully applied it to various programs manipulating AVL trees, (doubly-linked) lists, and combinations of both.
\end{abstract}

\section{Introduction}
The aim of shape analysis is to support software verification by discovering precise abstractions of the data structures in a program's heap.
For shape analyses to be effective, they need to track detailed information about the heap configurations arising during computations.
%
Although recent shape analyses have become quite potent~\cite{Abdulla2016verification,calcagno2011compositional,Chang2013modular,ferrara2012tval,reps2007shape}, discovering abstractions that go beyond structural shape properties remains far from fully solved.
For example, this is the case when considering \emph{balancedness properties} of data structures, such as the AVL property:
A full binary tree is an AVL tree if and only if for each of its inner nodes, the difference between the heights of its two subtrees is $-1$, $0$, or $1$.
In this setting, reasoning about constraints over lengths of paths or sizes of branches in a tree is required.
However, as already noted in \cite{Chang2013modular}, inference of shape-numeric invariants \enquote{is especially challenging and is not particularly well explored.}

We develop a shape analysis that is capable of inferring relational properties, such as balancedness, from a program and an intuitive data structure specification given by a \emph{graph grammar}.
Context-free graph grammars~\cite{habel1992hyperedge} have previously been successfully applied in shape analyses~\cite{heinen2015juggrnaut}.
They are, however, not expressive enough to capture typical relational properties of data structures. 
Hence, we lift the concept of \emph{indexed grammars} --- a classical extension of context-free string grammars due to Aho~\cite{aho1968indexed} --- to graph grammars. 
More concretely, we attach an \emph{index}, i.e. a finite sequence of symbols, to each nonterminal.
This information can then be accessed by the graph grammar to gain a fine-grained control over the applicable rules. 
For example, by using indices to represent the height of trees, a context-free graph grammar modeling binary trees can easily be lifted to a grammar representing \emph{balanced} binary trees.

One strength of indexed graph grammars is that they offer an intuitive formalism for specifying data structures without requiring deep knowledge about relational properties.
Furthermore, all key aspects of shape analysis (using the terminology of~\cite{reps2007shape}) have natural correspondences in the theoretically well-under\-stood domain of graph transformations:
\emph{Materialization}, an operation to partially concretize before performing a strong update of the heap, corresponds to the common notion of grammar derivations.
\emph{Concretization} then means exhaustively applying derivations.
Conversely, \emph{abstraction} (or canonicalization) coincides with applying inverse derivations as long as possible.
In particular, effective versions of the above operations can be derived automatically from a grammar through existing normal forms~\cite{jansen2011local}.
Finally, checking for \emph{subsumption} between two abstract states is an instance of the language inclusion problem for graph grammars.
While this problem is undecidable in general~\cite{bar1961formal}, we present a fragment of indexed graph grammars with a decidable language inclusion problem that is well-suited for shape analysis.

We implemented our shape analysis and successfully verified Java programs manipulating AVL trees, (doubly-linked) lists and combinations of both. 
Supplementary material to formalization and implementation is found in the Appendix.

\section{Informal Example}\label{sec:overview}
Our analysis is a standard forward abstract interpretation~\cite{cousot77} that approximates for each program location the set of reachable memory states.
It thus applies an abstract program semantics to elements of an abstract domain capturing the resulting sets until a fixed point is reached.
The analysis is parameterized by a user-supplied \emph{indexed hyperedge replacement grammar}: 
For any given grammar, we automatically derive an abstract program semantics from the concrete semantics of a programming language. 
Moreover, we obtain suitable abstraction and concretization functions.
In this section we take a brief tour through the essentials of our approach by means of an example.

\emph{Example program.}
We consider a procedure \texttt{searchAndSwap} (see Figure~\ref{fig:ov:code}) that takes an AVL tree \texttt{n} with back pointers and an integer value \texttt{key}.
It consists of two phases:
First, it performs a binary search in order to find a node in the tree with the given key (l. 9).
If such a node is found, it moves back to the root of the tree (l. 13). 
However, before moving up one level in the tree, the procedure swaps the two subtrees of the current node (l. 12).

\begin{figure}[t]
\centering
\scalebox{0.8}{
\begin{minipage}{0.27\textwidth}
\centering
\begin{align*}
       1\quad  & \texttt{class AVLTree} ~\{ \\
       2\quad  & \quad \texttt{AVLTree {\color{\lColor}left}}; \\
       3\quad  & \quad \texttt{AVLTree {\color{\rColor}right}}; \\
       4\quad  & \quad \texttt{AVLTree {\color{\pColor}parent}}; \\
       5\quad  & \quad \texttt{int key}; \\
       6\quad  & \quad \texttt{// \ldots} \\
       7\quad  & \}
\end{align*}
\end{minipage}\hfill
\begin{minipage}{0.73\textwidth}
\centering
\begin{align*}
       \qquad 8\quad & \texttt{void searchAndSwap(AVLTree n, int key)}\, \{  \\
       9\quad  & \quad \Ass{\texttt{n}}{\texttt{binarySearch(n, key)}}; \\
       10\quad  & \quad \WHILE{\NEquals{\texttt{n}}{\Null} \,\texttt{\&\&}\, \NEquals{\FP{\texttt{n}}}{\Null}} \{ \\
       11\quad  & \qquad \texttt{// swap subtrees of n} \\
       12\quad  & \qquad \texttt{AVLTree}\,\Ass{\texttt{t}}{\FL{\texttt{n}}};~\Ass{\FL{\texttt{n}}}{\FR{\texttt{n}}};
                       \Ass{\FR{\texttt{n}}}{\texttt{t}}; \\
      13\quad   & \qquad \Ass{\texttt{t}}{\Null}; \Ass{\texttt{n}}{\FP{\texttt{n}}}; \\
      14\quad   & \quad \} \\
      15\quad   & \}
\end{align*}
\end{minipage}
}
\caption{Essential fields of class \texttt{AVLTree} and example code.}
\label{fig:ov:code}
\end{figure}
%
%
%
%
%
\emph{Abstract domain.}
We assume a storeless model that is agnostic of concrete memory addresses.
Memory states are then naturally modeled as graphs --- more precisely \emph{indexed heap configurations} (IHC) (Section~\ref{sec:excursus}).
That is, an edge may be connected to an arbitrary number of nodes and is additionally labeled with an index that indicates, for instance, the height of a tree.
Consider the IHC depicted in Figure~\ref{fig:ov:graph}: 
A \emph{node} (drawn as a circle) either represents an object or a literal, such as $\Null$, $\True$, $\False$, etc.
The black circle denotes the special location $\Null$.\footnote{We often draw multiple black circles, but they all correspond to the same location.}
Pointers between objects are drawn as directed edges between two nodes that are drawn to indicate the corresponding field of its source object ($\Left$ (dashed), $\Right$ (dotted), and $\Parent$ (solid) for AVL trees). 
For example, the $\Parent$ pointer of the topmost node in Figure~\ref{fig:ov:graph} points to $\Null$.
\begin{wrapfigure}{r}{0.2\textwidth}
\vspace*{-3em}
\scalebox{0.8}{
\input{pics/ov_graph}
}
\caption{An IHC}
\label{fig:ov:graph}
\vspace*{-2em}
\end{wrapfigure}
%
Furthermore, IHCs contain \emph{program variables} and \emph{nonterminal edges}.
Program variables are drawn as diamonds that are labeled with the variable name and are attached to the unique node representing the value of the variable.
Hence, variable $n$ points to the rightmost node in Figure~\ref{fig:ov:graph}.
Nonterminal edges model a set of abstracted heap shapes, such as linked lists or balanced trees.
They are drawn as gray boxes and attached to one or more nodes.
Figure~\ref{fig:ov:graph} contains two of these edges. 
Their label, $B$, indicates that both model a set of balanced binary trees.
Further, their \emph{indices}, $X$ and $\Z X$, denote that they model balanced binary trees of height $X$ and $X + 1$, respectively, where $X$ stands for an arbitrary non-negative value.
Hence, the IHC in Figure~\ref{fig:ov:graph} models the set of all balanced binary trees with back pointers in which the height of the right subtree of the root is the height of its left subtree plus one.
Moreover, variable $n$ points to the right child of the root.

\emph{Abstraction and Concretization.}
The set of heaps described by an IHC is determined by an \emph{indexed hyperedge replacement grammar} whose rules map nonterminal edges to an IHC.
An example of a rule is provided in Figure~\ref{fig:ov:con} (inside the gray box; above step (1)).
Its left-hand side is $(B,\Z\iVar)$, where $\iVar$ is a variable.
The rule allows to replace any edge that is labeled with $B$ and whose index starts with an $\Z$ by the IHC below.
In that case, variable $\iVar$ is substituted by the remainder of the index of the replaced hyperedge.
The IHC on the rule's right-hand side contains two \emph{external nodes} (labeled $1$ and $2$) that indicate how two IHCs are glued together when replacing a hyperedge (Section~\ref{sec:excursus}).


\begin{figure}[t]
\begin{center}
\input{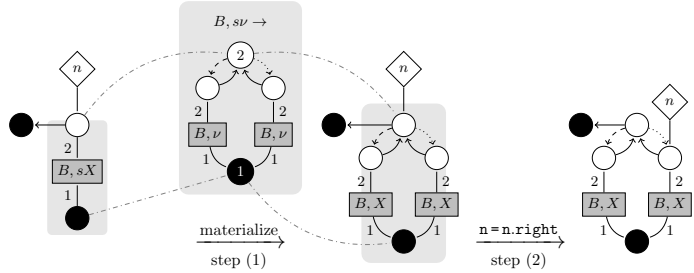}
\end{center}
\caption{Materialization and a possible execution of the binary search.}
\label{fig:ov:con}
\end{figure}

\emph{Example execution.}
Let us assume we are given a suitable grammar in which nonterminal $B$ represents balanced binary trees and index $\Z X$ stands for a height of $X + 1$.
%
We consider one 
execution sequence in detail.
The individual execution steps are illustrated in Figures~\ref{fig:ov:con},~\ref{fig:ov:switch}, and~\ref{fig:ov:abs}, respectively. 
Notice that the full analysis explores all abstract executions. 

\emph{Step (1).}
Starting with the leftmost IHC in Figure~\ref{fig:ov:con}, we first execute a binary search 
(Figure~\ref{fig:ov:code}, l. 9).
Assuming that the searched key is not at the root, we move to the children of $\texttt{n}$. 
Since these are currently hidden in the hyperedge labeled with $(B,sX)$, we apply \emph{materialization}~\cite{sagiv99parametric} (partial concretization). 
For our analysis, materialization corresponds to forward derivations using the supplied graph grammar, i.e. we replace an edge by an IHC according to a rule of the grammar.
Here, we used the rule above step (1) in Figure~\ref{fig:ov:con}.
To apply this rule, we first remove the original hyperedge labeled $(B,sX)$.
After that we paste the graph belonging to the rule into the original graph.
Finally, we identify the nodes originally attached to the removed hyperedge with the external nodes of the rule (as indicated by gray dashed and dotted lines in Figure~\ref{fig:ov:con}). 

\emph{Step (2).}
After materialization, executing one step of the concrete program semantics amounts to a simple graph transformation (moving variable \texttt{n} to a child).
To keep the example small, 
assume the binary search has already explored the left subtree without finding the key.
It thus returned to the root and the next step is to move variable $\texttt{n}$ to its right child.
That is, we execute 
$\Ass{\texttt{n}}{\FR{\texttt{n}}}$.
This leads to the rightmost graph depicted in Figure~\ref{fig:ov:con}.
In our example execution, we assume $\texttt{n}$ now carries the searched key, i.e. $\F{\texttt{n}}{\texttt{key}}$ equals $\texttt{key}$.
Hence, the binary search returns the current position of $\texttt{n}$ and we move to the $\texttt{while}$-loop of our example program (Figure~\ref{fig:ov:code}, l. 10).
Since neither variable $\texttt{n}$ is attached to $\Null$ nor its $\Parent$ pointer points to $\Null$, we enter the loop.

\begin{figure}[t]
\begin{center}
\input{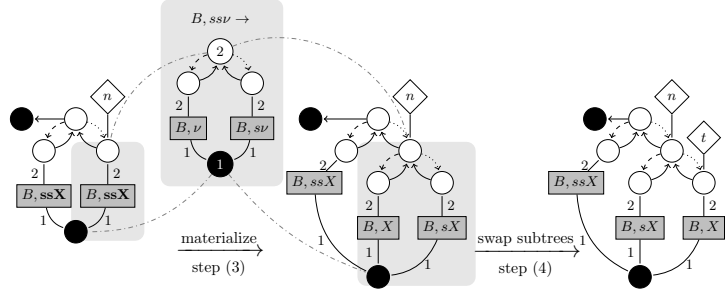}
\end{center}
\caption{Index materialization and swapping subtrees.}
\label{fig:ov:switch}
\end{figure}

\emph{Step (3).}
Before we can climb up the tree to the root again, we have to swap the subtrees of $\texttt{n}$ (Figure~\ref{fig:ov:code}, l. 12).
Again, these are hidden in a hyperedge labeled with $(B,X)$, i.e. we have to materialize again. 
As part of the example execution, we apply the rule in Figure~\ref{fig:ov:switch} (above step (3)).
However, this rule requires the index of a hyperedge to be of the form $\Z\Z \iVar$.
Intuitively, this means the rule models balanced trees of height at least two.
Since $X$ is a placeholder for trees of arbitrary height, we apply \emph{index materialization} to the IHC first.
That is, we replace $X$ by $\Z\Z X$ in all hyperedges\footnote{Again, note that we consider a single execution path in this example. The full analysis also explores the cases in which $X$ is substituted by $\iBot$ and $\Z \iBot$.}
and move to the leftmost hypergraph in Figure~\ref{fig:ov:switch}.
After that, we apply materialization as illustrated in the third step. 


\emph{Step (4).}
We apply the concrete semantics to execute a sequence of assignments in order to swap the left and right subtree of \texttt{n}
(Figure~\ref{fig:ov:code}, l. 12).
This results in the rightmost IHC of Figure~\ref{fig:ov:switch}, in which variable \texttt{t} has not been set to $\Null$ yet.
After executing the remaining two assignments, i.e.  $\Ass{\texttt{t}}{\Null}$ and $\Ass{\texttt{n}}{\FP{\texttt{n}}}$, we end up in the leftmost IHC in Figure~\ref{fig:ov:abs}.

Notice that both the abstract semantics as well as materialization are derived automatically from the grammar and the concrete program semantics (Sections~\ref{sec:excursus} and \ref{sec:analysis}).
In particular, materialization corresponds to forward derivations using the grammar. 
Analogously, the abstraction function 
corresponds to applying backward derivations.
Each occurrence of an IHC used as the right-hand side of a grammar rule is replaced by a hyperedge labeled with the rule's left-hand side.

\emph{Step (5).} 
After executing 
$\Ass{\texttt{n}}{\FP{\texttt{n}}}$ (Figure~\ref{fig:ov:code}, l. 13), 
abstracted is performed 
before moving on to the next loop iteration.
We abstract using a rule 
symmetric to the one applied in step (3) for materialization. 
This corresponds to first detecting the IHC in the rule as a subgraph of the given IHC.
This subgraph is deleted except for those nodes identified with the external nodes (labeled by numbers) of the rule graph (see gray dash-dotted lines in Figure~\ref{fig:ov:abs}). 
Then a hyperedge attached to the latter nodes is added to the remaining IHC.

\begin{figure}[t]
\begin{center}
\input{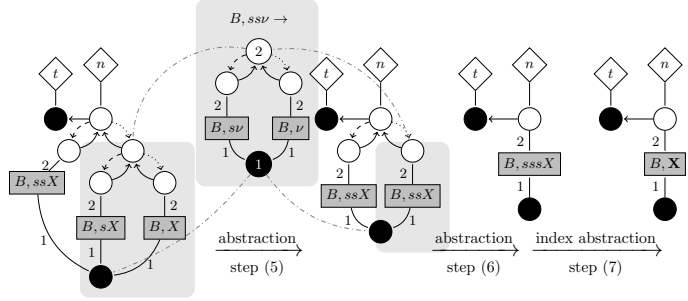}
\end{center}
\caption{Graph-based abstraction and index abstraction.}
\label{fig:ov:abs}
\end{figure}

\emph{Step (6).} 
The IHC obtained after step (5) can be further abstracted. 
This time, we employ the rule that has been applied for materialization first (Figure~\ref{fig:ov:con}, above step (1)).
The resulting graph is found in Figure~\ref{fig:ov:abs} next to step (6).
Note that the indices of both hyperedges to be abstracted are $\Z\Z X$ whereas the rule used for abstraction contains hyperedges with indices $\iVar$.
The variable $\iVar$ is used as a placeholder to restore the original indices after the replacement.
The resulting hypergraph (Figure~\ref{fig:ov:abs} following step (6)) contains a single hyperedge labeled $(B,sssX)$.
Hence, the result of our example execution is a balanced binary tree (of height at least three) again.

\emph{Step (7).}
As a final operation, we apply the converse of index materialization in step (3): index abstraction.
For this purpose, we replace $\Z\Z\Z X$ by $X$, i.e.\ we generalize from trees of height at least three to trees of arbitrary height.
Proceeding with the analysis, we evaluate the loop guard (Figure~\ref{fig:ov:code}, l. 10) to $\False$, because $\FP{\texttt{n}}$ equals $\Null$.
Hence, the analysis terminates this branch of its execution with a final hypergraph that covers the initial one.
The problem of checking whether a hypergraph covers another one is addressed in Section~\ref{sec:backwardconfluent}.

\section{Program States and Indexed Grammars}\label{sec:excursus}
As outlined in Section~\ref{sec:overview}, it is intuitive to model heaps as graphs.
In this section, we formalize heap configurations as a model for program states and their semantics in terms of a graph grammar.
These grammars guide concretization and abstraction in our analysis, which is presented subsequently in Section~\ref{sec:analysis}.

\subsection{Program States}

To set the stage for our analysis, we consider program states to consist of a heap and a stack.
We assume the heap to contain records with a finite number of \emph{reference fields} that are collected in $\Fields$. 
Apart from the heap, a program state is equipped with a \emph{stack} mapping program variables in $\Var$ to records.

Furthermore, our abstract domain equips graphs with \emph{nonterminal hyperedges} that act as abstract placeholders for sets of graphs, e.g. all (balanced) binary trees.
These hyperedges are labeled with a \emph{nonterminal} taken from a finite set $N$ and an \emph{index} taken from a finite set $\I$, respectively.
Throughout this paper, we fix a set $\Types = \Fields \cup \Var \cup N$.
Every element of $\Types$ is ranked by a function $\Rank : \Types \to \Nats$, where fields always have rank two, i.e.\ $\Rank(\Fields)= \{2\}$ and variables always have rank one, i.e.\ $\Rank(\Var) = \{1\}$, respectively.
Program states are then formally modeled as follows: 

\begin{definition}
\label{def:indexed-hg}
An \emph{indexed heap configuration} (IHC for short) is defined as a tuple
   $\iH = (\V,\E,\Lab,\Att,\Ind,\Ext)$, where
   \begin{itemize}
     \item $\V$ and $\E$ are finite sets of \emph{nodes} and \emph{hyperedges}, respectively,
     \item $\Lab : E \to \Types$ is a \emph{hyperedge labeling function},
     \item $\Att : \E \to \V^{*}$ maps each edge to a sequence of \emph{attached nodes} that respects 
             the rank of hyperedge labels, i.e. for all $e \in \E$, we have $\Rank(\Lab(e)) = |\Att(e)|$.
   \item $\Ind : E \to \I^{+}$ assigns a non-empty \emph{index sequence} to each edge in $\E$, and
     \item $\Ext \in \V^{+}$ is a repetition-free sequence of \emph{external nodes}.\footnote{External nodes are needed to define the semantics of nonterminal edges.} 
    \end{itemize}
\end{definition}

Throughout this paper, we do not distinguish between the terms graph and hypergraph nor between edge and hyperedge.
Furthermore, we refer to the components of a graph $\iH$ by $\V_{\iH}$, $\E_{\iH}$, etc.
If an edge $e$ is attached to exactly two nodes, say $\Att(e) = uv$, we interpret $e$ as a directed edge from node $u$ to node $v$. 
Notice that all graphs in Section~\ref{sec:overview} are examples of IHCs.

To simplify the technical development, we impose a few sanity conditions on IHCs:
We require that (1) every variable $x \in \Var$ occurs at most once in $\iH$ and (2) for every field $f \in \Fields$ every node has at most one outgoing edge $e$ labeled with $f$ (recall that $\Rank(f) = 2$).
The special location $\Null$ is treated as a global variable.
Hence, we assume a unique node $\vNull$ representing $\Null$ which is the first external node and the first node attached to every nonterminal edge.\footnote{I.e., $\vNull = \Ext(1)$ and for each $e \in E$ with $\Lab(e) \in N$, we have $\Att(e)(1) = \vNull$.}

\subsection{Indexed Grammars}

The semantics of edges labeled with a nonterminal, is specified by an indexed graph grammar --- an extension of context-free graph grammars. 
As it is common in graph rewriting, we do not distinguish between isomorphic graphs.
Thus, \emph{all sets of graphs in this paper are to be understood up to isomorphism.}\footnote{A formal definition of graph isomorphism is found in Appendix~\ref{app:excursus}.}

\begin{definition}
\label{def:ihrg}
Let $\iVar$ be a dedicated \emph{index variable} and $I' = I \cup \{\iVar\}$ be the set of index symbols.
An \emph{indexed hyperedge replacement grammar} (IG) is a finite set of rules $\iG$ of the form
$
  \Rule{X,\sigma}{\iH}
$ 
mapping a nonterminal $X \in N$ and an index $\sigma \in \I^{*} (I \cup \{\iVar\})$ to an IHC $\iH$ such that $\Rank(X) = |\Ext_{\iH}|$.
Moreover, if $\sigma$ does not contain the variable $\iVar$ then $\iH$ does not contain $\iVar$ either, i.e.
$\Ind_{\iH}(\E_{\iH}) \subseteq \I^{+}$.
\end{definition}

\begin{figure}[t]
    \center
    \scalebox{0.7}{
    \input{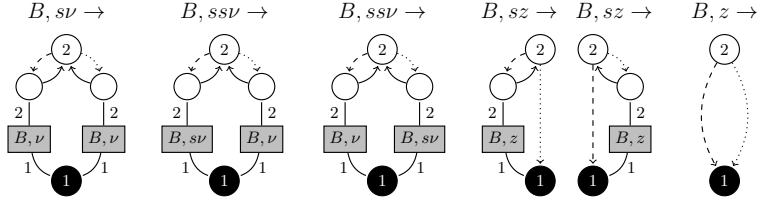}
    }
    \caption{An indexed hyperedge replacement grammar for balanced binary trees}
    \label{fig:hrg_balanced_trees}
\end{figure}


\begin{example}\label{ex:ihrg}
Figure~\ref{fig:hrg_balanced_trees} depicts an IG $\iG$ 
with  
six rules that each map to an IHC whose first external node is $\Null$ and whose second external node is the root of a tree-like graph.
Indices of edges not labeled with $\B$ are omitted for readability. 
%
%
\end{example}

The sets of graphs modeled by IGs are defined similarly to languages of context-free word grammars (CFG) in which a nonterminal is replaced by a finite string: 
An IG derivation replaces an edge, say $e$, that is labeled with a nonterminal by a finite graph, say $\iK$.
However, since arbitrarily many nodes may be attached to edge $e$, we have to clarify how the original graph and $\iK$ are glued together.
Hence, we identify each node attached to edge $e$ with an external node of $\iK$ (according to their position in both sequences).
Formally,

\begin{definition}
        \label{def:ihr}
  Let $\iH,\iK$ be IHCs with pairwise disjoint sets of nodes and edges.
  Moreover, let $e \in \E_{\iH}$ be an edge with $\Rank(\Lab_{\E_{\iH}}(e)) = |\Ext_{\iK}|$.
  Then the \emph{replacement} of $e$ in $\iH$ by $\iK$
  is given by $\Replace{\iH}{e}{\iK} = (\V,\E,\Att,\Lab,\Ind,\Ext)$, where
  \begin{align*}
          \V ~=~ & \V_{\iH} ~\cup~ \left( \V_{\iK} \setminus \Ext_{\iK} \right) 
                 &
          \E ~=~ & \underbrace{\left( \E_{\iH} \setminus \{ e \} \right)}_{=\,\E'} ~\cup~ \E_{\iK} \\
          \Lab ~=~ & \left(\Proj{\Lab_{\iH}}{\E'}\right) ~\cup~ \Lab_{\iK}
                   &
          \Ind ~=~ & \left(\Proj{\Ind_{\iH}}{\E'}\right) ~\cup~ \Ind_{\iK} \\
          \Att ~=~ & \left(\Proj{\Att_{\iH}}{\E'}\right) ~\cup~ \left( \Att_{\iK} \FComp \textit{mod} \right)
                   &
          \Ext ~=~ & \Ext_{\iH}
  \end{align*}
  where $\textit{mod}$ replaces each external node by the corresponding node attached to $e$.\footnote{$\Proj{f}{M}$ denotes the restriction of function $f$ to domain $M$ and $(f \FComp g)(s) = g(f(s))$. Moreover, function 
  $\textit{mod} = \{ \Elem{\Ext_{\iK}}{k} \mapsto \Elem{\Att_{\iH}(e)}{k} ~|~ 1 \leq k \leq |\Ext_{\iK}| \} \cup \{ v \mapsto v ~|~ v \in \V \setminus \Ext_{\iK} \}$ is lifted to sequences of nodes by pointwise application.
  }
\end{definition}

The above is the standard definition of hyperedge replacement in which indices and edge labels are treated the same (cf.~\cite{habel1992hyperedge}).
It is then tempting to define that an IG $\iG$ derives $\iK$ from $\iH$ if and only if 
there exists an edge $e \in \E_{\iH}$ and a rule $(\Rule{\Lab_{\iH}(e),\Ind_{\iH}(e)}{\iR}) \in \iG$ such that
$\iK$ is isomorphic to $\Replace{\iH}{e}{\iR}$.
However, this notion is too weak to model balanced trees.
In particular, since an index is treated as just another label, we cannot apply a derivation if the index of an edge does not exactly match an index on the left-hand side of an IG rule.

Instead, we use 
a finite prefix of indices in derivations and hide the remainder in variable $\iVar$. 
For example, assume an IG contains a rule $\Rule{\B,\Z\Z\iVar}{\iR}$.
Given an edge with label $\B$ and index $\sigma = \Z\Z\Z\iBot$, an IG derivation may then hide $\Z\iBot$ in $\iVar$.
The resulting index is $\Z\Z\iVar$ and a derivation as defined naively above is possible.
Finally,
all occurrences of 
$\iVar$ are replaced 
by the hidden suffix $\Z\iBot$ again. 

To formalize indexed derivations, two auxiliary definitions are needed:
Given a set $M \subseteq \Types$, we write $\E_{\iH}^{M}$ to refer to all edges of $\iH$ that are labeled with a symbol in $M$, i.e. $\E_{\iH}^{M} = \{ e \in \E_{\iH} ~|~ \Lab_{\iH}(e) \in M \}$.
We write $\Replace{\iH}{\iVar}{\rho}$ to replace all occurrences of $\iVar$ in (the index function $\Ind$ of) $\iH$ by $\rho$.\footnote{$\Replace{\iH}{\iVar}{\rho} = (\V_{\iH},\E_{\iH},\Att_{\iH},\Lab_{\iH},\Ind,\Ext_{\iH})$ with $\Ind = \{ \Replace{\Ind_{\iH}(e)}{\iVar}{\rho} ~|~ e \in \E_{\iH} \}$.}

\begin{definition}
\label{def:derive}
Let $\iG$ be an IG and $\iH,\iK$ be IHCs. 
Then $\iG$ \emph{directly derives} $\iK$ from $\iH$, written $\iH \iDerive{\iG} \iK$,
if and only if either
\begin{itemize}
  \item there exists a rule $(\Rule{X,\sigma}{\iR}) \in \iG$ and an edge $e \in \E_{\iH}^{\{X\}}$ such that
        $\Ind_{\iH}(e) = \sigma$ and $\iK$ is isomorphic to $\Replace{\iH}{e}{\iR}$, or 
  \item there exists a rule $(\Rule{X,\sigma\iVar}{\iR}) \in \iG$, an edge $e \in \E_{\iH}^{\{X\}}$, and 
          a sequence $\rho \in \I^{+}$ such that $\Ind_{\iH}(e) = \sigma \rho$ and $\iK$ is isomorphic to $\Replace{\iH}{e}{\Replace{\iR}{\iVar}{\rho}}$.
\end{itemize}
The reflexive, transitive closure of $\iDerive{\iG}$ is denoted by $\iDerive{G}^{*}$.
The inverse of $\iDerive{\iG}$ is given by $\rDerive{\iG}$. 
Finally, $\iH \nrDerive{\iG}$ iff there exists no $\iK$ such that $\iH \rDerive{\iG} \iK$.
\end{definition}


The \emph{language} of an IG and an IHC $\iH$ is the set of all graphs that can be derived from $\iH$ and that do not contain nonterminals.
Conversely, the inverse language of $\iH$ is obtained by exhaustively applying inverse derivations to $\iH$.

\begin{definition}\label{def:language}
The \emph{language} $\iL{\iG}$ and the \emph{inverse language} $\rL{\iG}$ of IG $\iG$ are given by the following functions mapping indexed graphs to sets of indexed graphs:
\begin{align*}
        \iL{\iG}(\iH) ~=~ & \{ \iK ~|~ \iH \iDerive{\iG}^{*} \iK ~\text{and}~ \E_{\iK}^{N} = \emptyset\},
        ~\text{and} \\
        \rL{\iG}(\iH) ~=~ & \{ \iK ~|~ \iH \rDerive{\iG}^{*} \iK ~\text{and}~ \iK \nrDerive{\iG} \}.
\end{align*}
\end{definition}

For instance, the language of the IG in Figure~\ref{fig:hrg_balanced_trees} for an IHC consisting of one edge labeled with $B,ssz$ is the set of all balanced binary trees of height two.

To ensure existence of inverse languages and thus termination of abstraction,
we assume that all rules of an IG $\iG$ are \emph{increasing}, i.e. for each rule $(\Rule{X,\sigma}{\iH}) \in \iG$ it holds that $|\V_{\iH}| + |\E_{\iH}| > \Rank(X)+1$.
As an example, notice that all rules of the IG in Figure~\ref{fig:hrg_balanced_trees} are increasing.
This amounts to a syntactic check on all rules that is easily discharged automatically.
We conclude our introduction of IGs with a collection of useful properties.

\begin{theorem}
\label{thm:ig-props}
  Let $\iG$ be an IG and $\iH$ be an IHC over $N$ and $\I$. Then: 
  \begin{enumerate}
    \item $\iH \iDerive{\iG}^{*} \iK$ implies $\Lang{\iG}{\iK} \subseteq \Lang{\iG}{\iH}$.
    \item $\Lang{\iG}{\iH} = 
            \begin{cases}
                    \{ \iH \} &~\text{if}~ \E_{\iH}^{N} = \emptyset \\ 
                \bigcup_{\iH \iDerive{\iG} \iK} \Lang{\iG}{\iK} &~\text{otherwise.}
            \end{cases}$
    \item It is decidable whether $\Lang{\iG}{\iH} = \emptyset$ holds.
    \item The inverse language $\iLang{\iG}{\iH}$ of an increasing IG $\iG$ is non-empty and finite.
  \end{enumerate}
\end{theorem}


The first two properties are crucial for proving our analysis sound.
The remaining properties ensure that we can construct well-defined (inverse) languages.

%

\section{Abstract Domain}\label{sec:analysis}
%
%
%
Our analysis  is a typical forward abstract interpretation~\cite{cousot92} that is parameterized by a user-supplied IG $\iG$.
Its concrete domain consists of all IHCs without nonterminals.
The abstract domain contains all IHCs to which no inverse IG derivation is applicable.
The order of our abstract domain is language inclusion.
Concretization $\con{\iG}$ and abstraction $\abs{\iG}$ 
correspond to computing the language and the inverse language of $\iG$, respectively.
Our setting is summarized in Figure~\ref{fig:domains}. 

\begin{figure}[t]
        \begin{tabular}{l@{~}l@{\quad}ll}
                Concrete domain: & $(\DomCon = \PS{\underbrace{\Lang{\iG}{\IHC}}_{\text{concrete IHCs}}}, \subseteq)$
                         & Concretization: & $\con{\iG} ~=~ \iL{\iG}$
\\
Abstract domain: & $(\DomAbs = \PS{\underbrace{\iLang{\iG}{\IHC}}_{\text{fully abstract IHCs}}}, \LOrder{})$,
                 &
Abstraction: & $\abs{\iG} ~=~ \rL{\iG}$
\end{tabular}
\caption{
Concretization, abstraction and the respective domains for a given IG $\iG$. 
Here, $\LOrder{}$ is given by $\iH \LOrder{} \iK$ iff $\con{\iG}(\iH) \subseteq \con{\iG}(\iK)$.
$\PS{M}$ is the powerset of $M$.
This setting yields a Galois connection for backward confluent IGs (cf. Section~\ref{sec:backwardconfluent}).
}
\label{fig:domains}
\end{figure}

The concrete semantics of common imperative programs amounts to straightforward graph transformations.
Let us assume that $\Progs$ is the set of all programs.
Moreover, assume the concrete semantics of each program $P \in \Progs$ is given by a (partial) function $\PSem{P} : \IHC \to \IHC$ that captures the effect of executing $P$ on an IHC.\footnote{We defined $\PSem{.}$ inductively for a Java-like language. Furthermore, we proved locality for all $\Progs$-programs. Details are found in Appendix~\ref{app:sec-concrete-semantics}.} 
For example, step (2) in Section~\ref{sec:overview} computes $\PSem{\Ass{\texttt{n}}{\FR{\texttt{n}}}}$.

As is standard, our analysis performs a fixed-point iteration of the abstract semantics that overapproximates the concrete semantics.
Following the terminology of~\cite{reps2007shape}, our abstract semantics consists of three phases: materialization, execution of the concrete semantics, and canonicalization.
That is, our abstract semantics is a function of the form
$
\PASem{.} \,:\, \Progs \to \DomAbs \to \DomAbs 
$ 
that is defined inductively on the structure of programs. 
In particular, for an atomic program $P \in \Progs$, we have
$\PASem{P} = \Mater{\iG}{P} \FComp \PSem{P} \FComp \Canon{\iG}{P}$.\footnote{$f \FComp g$ denotes sequential composition of $f$ and $g$, i.e. $(f \FComp g)(s) = g(f(s))$.}
The inductive cases are straightforward (cf. Appendix~\ref{app:sec-abstract-semantics}).

Although materialization and canonicalization naturally depend on the user-provided grammar $\iG$, for readability we tacitly omit adding $\iG$ as a parameter. 
Materialization ensures applicability of the concrete semantics by partially concretizing an IHC.
%
It is thus a function 
$
    \Mater{\iG}{.} : \Progs \to \IHC \to \PSF{\IHC}
$ 
that, for a given program, maps an IHC to a finite set of IHCs.
Intuitively, materialization applies derivations $\iDerive{\iG}$ 
until the concrete semantics can be applied (cf. Theorem~\ref{thm:ig-props}.2).
A detailed discussion of suitable materializations that are derived from a grammar $\iG$ is found in~\cite{heinen2015juggrnaut,jansen2011local}.
In this paper, we consider a sufficient condition to ensure soundness. 
\begin{definition}\label{prop:mater}
        For every atomic program $P \in \Progs$, we require a materialization function $\Mater{\iG}{.}$ such that
        $ 
        \con{\iG} \FComp \PSem{P} \FOrder \Mater{\iG}{P} \FComp \PSem{P} \FComp \con{\iG}. 
        $
\end{definition}

Here, $\FOrder$ denotes pointwise application of $\subseteq$.
Examples of applying materialization are provided in steps (1) and (3) of Section~\ref{sec:overview}.

Conversely to materialization, canonicalization takes a partially concretized program state and computes an abstract program state again.
It is thus a function of the form
$\Canon{\iG}{.} : \Progs ~\to~ \IHC ~\to~ \DomAbs$.
\begin{definition}\label{prop:canon}
   For every program $P \in \Progs$, we require a canonicalization function $\Canon{\iG}{.}$ such that $\con{\iG} ~\FOrder~ \Canon{\iG}{P} \FComp \con{\iG}$.
\end{definition}
By Theorem~\ref{thm:ig-props}(1), inverse IG derivations as well as the abstraction function $\abs{\iG}$ are suitable candidates for canonicalization.
Examples of applying canonicalization are provided in steps (5) and (6) of Section~\ref{sec:overview}.

%

Assuming suitable materialization and canonicalization functions as of Definitions~\ref{prop:mater} and~\ref{prop:canon},
our abstract semantics $\PASem{.}$ computes an overapproximation of the concrete semantics $\PSem{.}$ (detailed proofs are found in Appendix~\ref{app:sec-sound}): 

\begin{theorem}[Soundness]\label{thm:sound}
For all $\pP \in \Progs$, $\con{\iG} \FComp \PSem{P} ~\FOrder~ \PASem{P} \FComp \con{\iG}$.
\end{theorem}

The quality of our analysis depends, naturally, on the quality of the user-defined grammar.
That is, the better our grammar matches the data structures employed by a program, the more precise the results obtained from our analysis.
In particular, our analysis does not necessarily terminate.
For example, we cannot analyze a program working on doubly-linked lists if the user-supplied IG models trees only.
As usual, termination has to be ensured by some sort of widening.
In the simplest case, termination is achieved by fixing a maximal size of IHCs a priori.
Whenever an IHC exceeds the fixed size, the analysis stops.

%

\section{Backward Confluent IGs}\label{sec:backwardconfluent}
%
Two components of our analysis are particularly involved:
First, the inverse language of an IHC with respect to an IG has to be computed repeatedly during canonicalization, i.e.\ we have to exhaustively apply inverse IG derivations.
Applying inverse derivations in turn requires finding isomorphic subgraphs in an IHC that can be replaced by a hyperedge. 
Since the subgraph isomorphism problem is \textsc{NP}-complete~\cite{cook1971complexity}, canonicalization is expensive. 

Second, computing a fixed point requires us to check for language inclusion. 
However, the language inclusion problem for IGs is undecidable as it is already undecidable for context-free string grammars~\cite{bar1961formal}. 
Undecidability of inclusion
is common in the area of shape analysis, where supported data structures are either severely restricted to obtain decidability, or approximations are used.

We now discuss a subclass of IGs 
that addresses both problems:

\begin{definition}
\label{def:backwardconfluent}
An IG $\iG$ is \emph{backward confluent} iff 
for all IHCs 
$\iH$ the inverse language $\iLang{\iG}{\iH}$ is a singleton set, i.e. $|\iLang{\iG}{\iH}| = 1$.
\end{definition}

The definition of backward confluent IGs is, admittedly, rather semantics-driven.
In particular, it solves the problem of expensive canonicalizations directly: Since the inverse language of an IHC is unique it suffices to exhaustively apply inverse derivations 
instead of trying all possible combinations.
Fortunately, as shown in~\cite{plump2010checking}, backward confluence can be checked automatically. 
In particular, we constructed backward confluent IGs for singly- and doubly-linked, (a)cyclic lists, (balanced) trees (w/o back pointers), in-trees, lists of lists, and (in-)trees with linked leaves.
In general, however, the class of graph languages generated by backward confluent IGs is strictly smaller than the class of languages generated by arbitrary IGs.\footnote{A formal proof is found in Appendix~\ref{app:sec-backwardconfluent-expressiveness}.}

%



We now turn to our second desired property: a decidable inclusion problem.
This property relies on the observation that two IHCs $\iH,\iK$ that cannot be abstracted further, i.e. $\iH,\iK \nrDerive{\iG}$, 
are either isomorphic or have disjoint languages.

%
%
%
\begin{theorem}\label{thm:backwardconfluent-inclusion}
  Let $\iG$ be a backward confluent IG.
  Moreover, let $\iH, \iK \in \IHC$ such that $\iK \nrDerive{\iG}$. 
  Then it is decidable whether 
  $\Lang{\iG}{\iH} \subseteq \Lang{\iG}{\iK}$ holds.
\end{theorem}




To conclude this section, we remark that, for backward-confluent IGs, our concrete and abstract domain (cf.\ Figure~\ref{fig:domains}) form a Galois connection, i.e. our analysis falls within the classical setting of abstract interpretation~\cite{cousot77}.

\section{Global Index Abstraction}\label{sec:global}
The 
goal of our shape analysis is to enable reasoning about complex data structures, such as balanced binary trees.
However, we might encounter infinitely many IHCs that vary in their indices only, thus preventing termination (cf. steps (1) and step (6) in Section~\ref{sec:overview}).
Our abstraction is thus often too precise.

To capture that an IHC models balanced trees, however, it suffices to keep track of the \emph{differences} between indices:
Assume, for example, that a node has two subtrees specified by nonterminal edges with indices
$\Z\iBot$ and $\Z\Z\iBot$. If we replace these indices by $\Z\Z\iBot$ and $\Z\Z\Z\iBot$, the underlying trees remain balanced.

Hence, we propose an index 
abstraction on top of IG-based abstraction.
Intuitively, this abstraction removes a common suffix from all indices and replaces it by a placeholder.
Apart from balancedness, it is applicable to properties such as \enquote{all sublists in a list of lists have equal length}. 
The abstraction is again formalized by grammars; right-linear context-free word grammars (CFG) to be precise.
Thus, let $\I = \I_{N} \cup \I_{T}$ be a finite set of index symbols that is partitioned into a set of nonterminals $\I_{N}$ and a set of terminals $\I_{T}$ including the end-of-index symbol $\iBot$.
We call an index $\sigma \in \I^{+}$ \emph{well-formed} if $\sigma \in (\I_{T} \setminus \{\iBot\})^{*} (\I_{N} \cup \{\iBot\})$.
That is, a well-formed index always ends with a nonterminal or the end-of-index symbol $\iBot$.
Accordingly, an IHC is well-formed if all of its indices are.
We assume all indices --- including indices in CFG rules --- to be well-formed.
Hence, all considered CFGs are right-linear and thus generate regular languages.
We do not allow nonterminal index symbols in IGs, i.e. we
assume for each IG rule $\Rule{X,\sigma}{\iH}$ that $\Ind_{\iH}(\E_{\iH}^{N}) \subseteq \I_{T}^{*} \{\iBot, \iVar \}$, where $\iVar$ has been 
introduced in Definition~\ref{def:ihrg}.

To maintain relationships between indices, such as their difference, we require that all indices ending with the same nonterminal of an IHC are modified simultaneously.
This leads us to a notion of global derivations and 
global languages.

\begin{definition}\label{def:indexabstraction}
Let $\iH,\iK \in \IHC$.
A CFG $\iC$ \emph{globally derives} $\iK$ from $\iH$, written $\iH \iCon{\iC} \iK$, if and only if there exists a rule $(\Rule{X}{\tau}) \in \iC$ such that 
$\Ind_{\iH}(\E_{\iH}^{N}) \subseteq \I_{T}^{*} \I_{N}$ 
and $\iK$ is isomorphic to $\Replace{\iH}{X}{\tau}$, i.e. $\iH$ in which all occurrences of $X$ are replaced by $\tau$.
Again, $\iCon{\iC}^{*}$ the reflexive, transitive closure of $\iCon{\iC}$.
$\iAbs{\iC}$ denotes inverse derivations and $\niAbs{\iC}$ that no inverse derivation is possible.
\end{definition}
%
%
\begin{definition} \label{def:globallanguage}
The \emph{global language} and the \emph{inverse global language} of a right-linear CFG $\iC$ over $\I$ are given by:
\begin{align*}
\iGL{\iC} : \IHC \to \PS{\IHC},~ 
\iH \,\mapsto\, & \{ \iK ~|~ \iH \iCon{\iC}^{*} \iK ~\text{and}~ \Ind_{\iK}(\E^{N}_{\iK}) \subseteq \I_{T}^{+} \} \\
\rGL{\iC} : \IHC \to \PS{\IHC},~
\iH \,\mapsto\, & \{ \iK ~|~ \iH \iAbs{\iC}^{*} \iK ~\text{and}~ \iK \niAbs{\iC} \}
\end{align*}
\end{definition}




Global derivations enjoy the same properties as IG derivations (Theorem~\ref{thm:ig-props}).
These properties are crucial to ensure soundness and termination of abstraction.

%

To combine global derivations and IG derivations, we consider a new derivation relation of the form $(\iDerive{\iG} \cup \iCon{\iC})^{*}$. 
We can further simply this relation, because global derivations and IG derivations enjoy an orthogonality property:
%
\begin{theorem}
        $\iH \, (\iDerive{\iG} \cup \iCon{\iC})^{*} \, \iK ~\text{if and only if}~ \iH \, (\iCon{\iC}^{*} \FComp \iDerive{\iG}^{*}) \, \iK$.\footnote{For binary relations $R_1,R_2$, we set $R_1 \FComp R_2 = \{(u,w) ~|~ \exists v : (u,v) \in R_1, (v,w) \in R_2 \}$.}
\end{theorem}
Thus, for materialization, it suffices to first apply global derivations and then apply IG derivations.
Conversely, for abstraction, it suffices to first apply inverse IG derivations and then apply inverse global derivations.

It is then straightforward to refine our analysis from Section~\ref{sec:analysis} by using the above derivation relation (cf.\ Appendix~\ref{app:sec-cfg-analysis}).
To conclude this section, we remark that all results from Sections~\ref{sec:analysis} and~\ref{sec:backwardconfluent} can be lifted to the refined analysis.

\section{Implementation} \label{sec:impl}
We implemented our analysis in \textsc{Attestor}~\cite{attestorCav} to analyze Java programs. 
The source code and our experiments are available 
online.\footnote{\url{https://github.com/moves-rwth/attestor-examples/releases/tag/v0.3.5-SEFM2018}\newcounter{attestorfn}\setcounter{attestorfn}{\value{footnote}}}


\emph{Input.} 
\textsc{Attestor} supports a fragment of Java that 
includes recursive procedure calls, but no arithmetic. 
Apart from programs and grammars, linear temporal logic (LTL) specifications over execution paths can be supplied. 
Atomic propositions 
include heap shapes and reachability of variables
(cf.~\cite{esop2017}).

\emph{Output.}
\textsc{Attestor} generates a transition system in which each state consists of a program location and an IHC representing the abstract program state, i.e., a set of reachable heaps.
This state space can also be explored graphically.\footnote{A brief tutorial on using the tool is found in Appendix~\ref{app:sec-tutorial}.}
Collecting the IHCs of all states with the same program location then coincides with the result of the abstract semantics presented in Section~\ref{sec:analysis}.
Moreover, the tool applies LTL model-checking to verify provided LTL specifications.

\begin{table}[t]
\center
\scalebox{0.9}{
\begin{tabular}{lr|r||lr|r}
    \textbf{Program} & \textbf{Mem. Safety} & \textbf{Shape} & 
    \textbf{Program} & \textbf{Mem. Safety} & \textbf{Properties}  \\
    \hline\hline
    \multicolumn{3}{l||}{AVL trees with parent pointers} & 
    \multicolumn{3}{l}{Data structure traversals / other algorithms} \\
    \cline{1-6}
    \texttt{binary search}  & 0.089 & 0.153 & 
    \texttt{List of cyclic lists} & 0.115 & 0.115  \\
    \texttt{min. value}      & 0.128 & 0.204 & 
    \texttt{Tree (Lindstrom)}   & 0.084  & 11.60  \\
    \texttt{search and back} & 0.140 & 0.158 & 
    \texttt{Skip list}   & 0.117  & 0.117  \\
    \texttt{search and swap} & 0.823 & 1.106 &  
    \texttt{Tree (recursive)}   & 0.080  & 8.700  \\
    \texttt{rebalance} & 1.500 & 1.769 &  
    \texttt{Zip list (recursive)}      & 0.118 & 0.118  \\
    \texttt{insert} & 1.562 & 3.079 &  
    \texttt{DLL reversal}      & 0.054 & 0.126 \\
    \texttt{list to AVLTree}        & 1.784 & 1.892 &  
    \texttt{DLL insertion sort}      & 0.369 & 1.134 \\
    \hline
\end{tabular}
}
\medskip
\caption{
        An excerpt of our experimental results. Provided verification times are in seconds including model-checking. Verified properties include \emph{memory safety}, correct heap \emph{shape} (including balancedness), correct return values, every element has been accessed, and the input data structure coincides with the output data structure. 
        Column \emph{properties} provides the worst runtime for verified LTL specifications.
        The complete benchmark results are found in Appendix~\ref{app:sec-tutorial}.        
}
\label{tbl:benchmarks}
\end{table}

\emph{Experimental results.}
We evaluated our implementation against common challenging algorithms 
on various data structures and multiple LTL specifications.
The results are shown in Table~\ref{tbl:benchmarks}.
Experiments were performed on an Intel Core i7-5820K at 3.30GHz with the Java virtual machine limited to 2GB of RAM.
Program inputs covered all instances of the respective data structure through nonterminal edges for each employed data structure. 
Further details regarding individual case studies are provided in Appendix~\ref{app:sec-tutorial} and online.\footnotemark[\value{attestorfn}]
In particular, \texttt{list to AVLTree} traverses a singly-linked list while inserting each of its elements into an (initially empty) AVL tree including all rebalancing procedures. 
Our implementation successfully verifies that the result is a balanced binary tree and the list has been completely traversed.
This demonstrates that our analysis is capable of precisely reasoning about combinations of multiple data structures.

\section{Related Work} \label{sec:related}
%


\emph{Graph Transformations.}
Our work is an extension of an existing analysis based on context-free graph grammars~\cite{heinen2015juggrnaut}: 
From a theoretical perspective, IGs allow covering infinitely many context-free rules by a single nonterminal with an index variable. 
Covering infinitely many rules is essential when reasoning about relational properties, e.g. balancedness. 
From a practical perspective, our analysis is a standard forward abstract interpretation in contrast to previous approaches.
%

\emph{Separation Logic.}
The class of graphs described by context-free graph grammars is equivalent to a fragment of symbolic heap separation logic (SL)~\cite{jansen2014generating}. 
In contrast to SL, graph grammars give us access to a rich set of theoretical results from string and graph rewriting.
For example, the concept of IGs is derived from Aho's indexed string grammars~\cite{aho1968indexed}.
Moreover, the notion of backward confluence is well-studied in the context of graph rewriting (cf.~\cite{plump2010checking}) and provides us with a decidable criterion to discharge entailments (language inclusion).
\textsc{Hip/Sleek} uses SL enriched with arithmetic to specify size constraints on data structures (cf. \cite{Chin2012automated}).
Their focus is on program verification with user-supplied invariants.
In contrast, our approach 
synthesizes invariants automatically. 
Furthermore, we provide decidable criteria for good data structure specifications 
whereas \textsc{Hip/Sleek} relies on heuristics 
to discharge entailments. 

\emph{Static analysis.}
%
%
\cite{Chang2008relational,Chang2007shape} introduce a generic framework for relational inductive shape analysis based on user-supplied invariants. 
Applicability to red-black trees is demonstrated in an example, but not covered by experiments. 
%
In \cite{Abdulla2016verification}, forest automata are extended
by constraints between data elements associated with nodes of the heaps. 
The authors conjecture that their method generalizes to handle 
lengths of branches in a tree, which are needed 
to express balancedness properties. The details, however, are not worked out.
%
%
%
%
%
%
%
%

\section{Conclusion} \label{sec:conclusion}
We developed a shape analysis that is capable of proving certain relational properties of data structures, such as balancedness of AVL trees. 
Our analysis is parameterized by user-supplied indexed graph grammars --- a novel extension of context-free graph grammars.
We implemented our approach and successfully applied it to common algorithms on AVL trees, lists, and combinations thereof.

\bibliographystyle{splncs03}
\bibliography{bibliography}

\clearpage
\appendix
\section{Appendix}
The appendix contains missing proofs, detailed formalizations of the concrete and abstract semantics, and further details regarding the implementation.
It is structured as follows:

\begin{itemize}
    \item Appendix~\ref{app:sec-tutorial} is a brief tutorial explaining how \textsc{Attestor} is installed and executed. 
          In particular, we show how our experimental results (see Section~\ref{sec:impl}) can be reproduced.
          Moreover, we briefly describe how results can be graphically explored.
    \item Appendix~\ref{app:sec-concrete-semantics} formally defines a simple imperative programming language $\Progs$ together with its concrete semantics $\PSem{.}$ defined on indexed heap configurations (see Section~\ref{sec:analysis}).
    \item Appendix~\ref{app:excursus} provides the missing formal definition of isomorphic indexed heap configurations.
    \item Appendix~\ref{app:sec-abstract-semantics} formally defines the abstract semantics of $\Progs$ programs that was informally introduced in Section~\ref{sec:analysis}.
    \item Appendix~\ref{app:sec-ig-props} contains the proof of Theorem~\ref{thm:ig-props}, i.e. properties of indexed graph grammars.
    \item Appendix~\ref{app:sec-sound} formalizes that our analysis is \emph{sound} (Theorem~\ref{thm:sound}).
    \item Appendix~\ref{app:sec-local} formalizes that our analysis allows for \emph{local reasoning}. 
    \item Appendix~\ref{app:sec-backwardconfluent-expressiveness} formally shows that not every language of an IG can also be expressed by a backward confluent IG (see Section~\ref{sec:backwardconfluent}).
    \item Appendix~\ref{app:sec-backwardconfluent-inclusion} formally proves that language inclusion is decidable for backward confluent IGs, see Theorem~\ref{thm:backwardconfluent-inclusion}.
    \item Appendix~\ref{app:sec-cfg-props} contains the missing proofs showing that index abstraction and concretization preserves the desirable properties of IGs (see Section~\ref{sec:global}).
    \item Appendix~\ref{app:sec-cfg-analysis} formally defines the modified analysis that additionally uses index abstraction (see Section~\ref{sec:global}).
    \item Appendix~\ref{app:grammars} contains additional examples of indexed graph grammars.
\end{itemize}

In case of publication, the appendix will be made available online as a separate technical report.

\subsection{Tutorial: Reproducing Experimental Results with \textsc{Attestor}}\label{app:sec-tutorial}

This section is a brief tutorial on reproducing our experimental results.
Furthermore, a more detailed table with all benchmark results is found at the end of this tutorial.

\subsubsection{System Requirements} In order to use \textsc{Attestor}, please first make sure that the following software is installed:
\begin{itemize}
        \item Java JDK 1.8 (see \url{http://www.oracle.com/technetwork/java/javase/downloads/jdk8-downloads-2133151.html})
        \item Apache Maven (see \url{https://maven.apache.org/})
        \item Git (see \url{https://git-scm.com/})
\end{itemize}
Furthermore, notice that an active internet connection is required during installation as maven will automatically download additional required packages.

\subsubsection{Reproducing Experiments} 
A bundle of \textsc{Attestor} and our experiments is obtained as follows:
\begin{verbatim}
  git clone --branch v0.3.5-SEFM2018 \
  https://github.com/moves-rwth/attestor-examples.git
\end{verbatim}
We provide a shell script that automatically installs the bundle, executes all experiments and generates a latex document with the results.
To use the script, please run the following inside of the cloned repository:
\begin{verbatim}
  chmod +x run.sh
  ./run.sh
  pdflatex benchmark-results.tex
\end{verbatim}
Alternatively, e.g. if shell scripts cannot be executed on your operating system, 
all experiments can be executed using maven:
\begin{verbatim}
 mvn clean install exec:exec@run
\end{verbatim}
Notice that no latex document will be generated without using the shell script. 
All relevant data are displayed on the console though and are additionally exported to \texttt{benchmark-results.csv}.

\begin{sidewaysfigure}[p]
    \includegraphics[width=\textwidth]{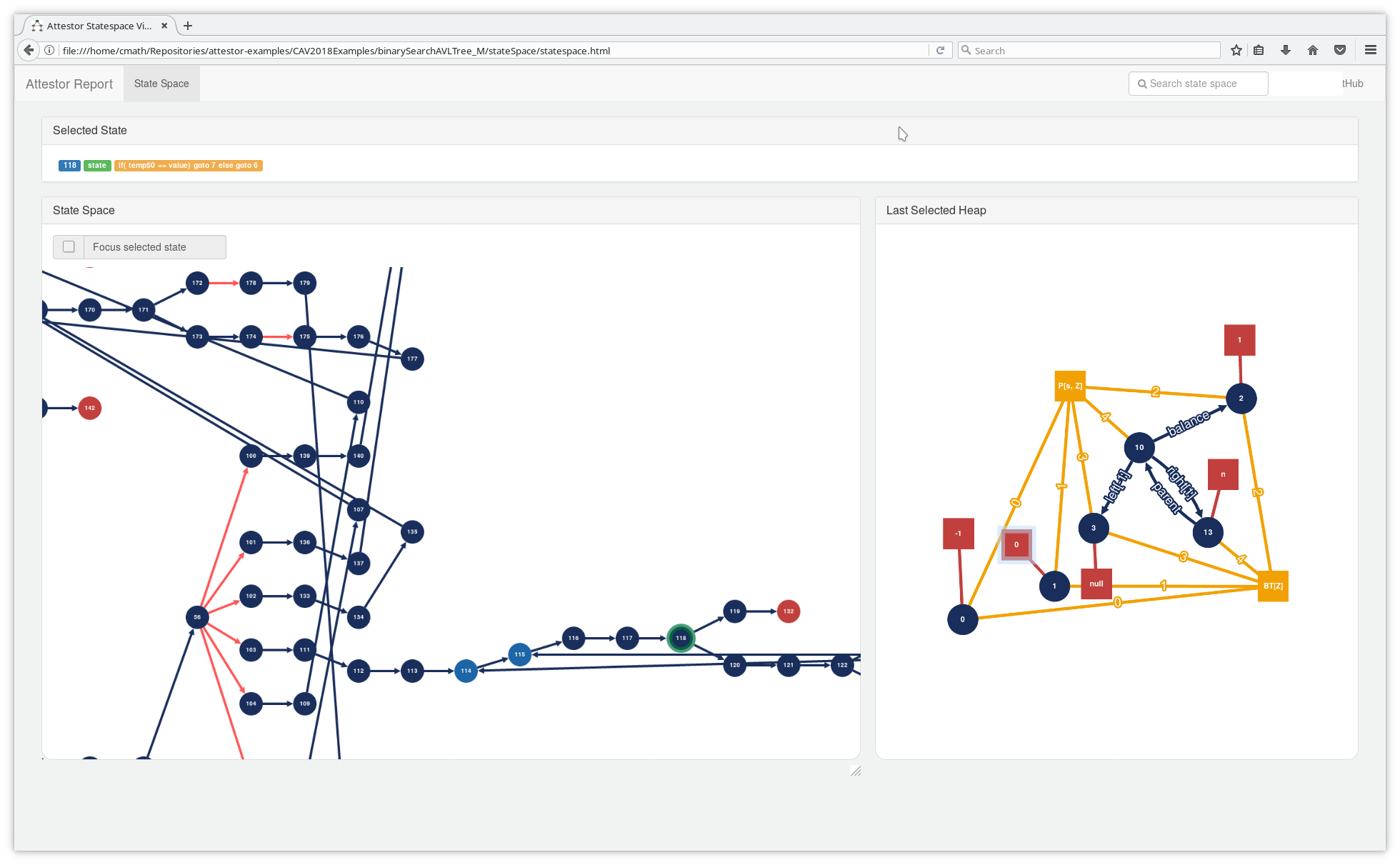}
    \caption{Screenshot of graphical state space exploration}
    \label{screenshot}
\end{sidewaysfigure}

\subsubsection{Graphical Exploration}
It is possible to generate a \enquote{report} to graphically explore generated state spaces.
To this end, install the \textsc{Attestor} bundle as described above. Then execute:
\begin{verbatim}
  mvn clean install exec:exec@runWithReport
\end{verbatim}
Notice that this will take considerably more time than just running all benchmarks.
For each benchmark, an additional directory will be created containing a website that allows to explore the state space in your web browser.
The names of created directories match the names used in the benchmark settings files.

For example, consider the AVL binary search benchmark.
Assuming you cloned the \texttt{attestor-examples} repository, the configuration file of this benchmark is located in
\begin{verbatim}
attestor-examples/AVLTree/configuration/settings/binary_search.json
\end{verbatim}
The corresponding website to graphically explore the state space is found in
\begin{verbatim}
  attestor-examples/AVLTree/binary_search/stateSpace
\end{verbatim}
To explore the state space, first start a web server by running (inside of the created directory)
\begin{verbatim}
  java -jar attestor-report
\end{verbatim}
On most operating systems, this should open your default web browser and display the report (otherwise start your web browser and go to \url{http://localhost:8080}).
Notice that you can also open \texttt{index.html} directly if you are using Firefox.
A screenshot of the graphical state space exploration (for a binary search in an AVL tree) is depicted in Figure~\ref{screenshot}.
The left pane depicts the state space. The right pane depicts the currently selected heap configuration.
Moreover, the topmost pane displays further information about the currently selected state, e.g. the corresponding program statement to be executed and its atomic propositions to be considered for model-checking.

\subsubsection{Installing Attestor}
The version of \textsc{Attestor} used in this paper is available on GitHub\footnote{\url{https://github.com/moves-rwth/attestor/releases/tag/v0.3.5-SEFM2018}}
and maven central.\footnote{\url{http://repo1.maven.org/maven2/de/rwth/i2/attestor/0.3.5-SEFM2018/}}
To install \textsc{Attestor}, please execute the following in your terminal:
\begin{verbatim}
  git clone --branch v0.3.5-SEFM2018 \
  https://github.com/moves-rwth/attestor.git
  mvn clean install
\end{verbatim}
A corresponding jar file including all dependencies will be generated in the target directory. Please confer \url{https://github.com/moves-rwth/attestor/wiki} for further details.
Example files are provided in the attestor-examples repository (see \enquote{Reproducing Experiments} above).

\subsubsection{Full Benchmark Results}

Table~\ref{table:full} below depicts the results for our full collection of benchmarks.
Each benchmark is marked with additional flags indicating the verified properties.
More precisely,
\begin{itemize}
        \item (M) means that we checked \emph{memory safety}.
        \item (S) means that we checked \emph{shape properties}, e.g. that a data structure is indeed an AVLTree.
        \item (C) means that we checked \emph{correctness properties}, e.g. that the head pointer is placed correctly upon termination.
        \item (V) means that we checked whether every node in the initial heap has been \emph{visited}, e.g. all elements of a list have been traversed.
        \item (N) means that we checked whether the \emph{neighbourhood} of every node in the initial data structure is the same upon termination, e.g.  the output data structure coincides with the input data structure.
        \item (X) means that a property is violated and we successfully constructed a non-spurious counterexample.
        \item (Y) means that a property is violated, but all counterexamples are spurious.
\end{itemize}
For each benchmark, we consider
\begin{itemize}
   \item the total number of generated states (\#States),
   \item state space generation time in seconds (SSG),
   \item model-checking time in seconds (MC),
   \item total verification time (including SSG and MC) in seconds (Verif.), and
   \item the total runtime (including parsing) in seconds (Total).
\end{itemize}


Further information about each individual benchmark is found in the examples repository (see \enquote{Reproducing Experiments}).

\begin{longtable}{|l|r|r|r|r|r|}
\caption{Full Benchmark Results}\label{table:full}
\endfirsthead
\multicolumn{6}{@{}l}{\ldots continued}\\[1ex]\hline
\bfseries Benchmark & \bfseries \#States & \bfseries SSG & \bfseries MC & \bfseries Verif. & \bfseries Total
\\ \hline \hline
\endhead 
\hline
\endfoot
        \hline
        \bfseries Benchmark & \bfseries \#States & \bfseries SSG & \bfseries MC & \bfseries Verif. & \bfseries Total
        \\ \hline \hline
        \begin{filecontents*}{benchmark-results.csv}
name,property,states,generation,mc,verification,total
AVL: Binary search,M,192,0.088,0.000,0.089,0.958
AVL: Binary search,M-C,192,0.089,0.016,0.106,0.955
AVL: Binary search,M-S,192,0.135,0.014,0.153,1.003
AVL: Binary search,M-S (final),192,0.133,0.015,0.153,1.007
AVL: Insertion,M,6120,1.561,0.000,1.562,2.399
AVL: Insertion,M-S,10386,2.711,0.363,3.079,3.926
AVL: Left-Right Rotation,M,1539,0.685,0.000,0.686,1.520
AVL: Left Rotation,M,201,0.144,0.000,0.145,0.983
AVL: List to AVLTree,M,7166,1.751,0.000,1.752,2.591
AVL: List to AVLTree,M-C,7166,1.754,0.025,1.779,2.634
AVL: List to AVLTree,M-S,7166,1.684,0.020,1.709,2.564
AVL: maximal value,M,71,0.057,0.000,0.057,0.887
AVL: minimal value,M,220,0.128,0.000,0.128,0.954
AVL: minimal value,M-C,220,0.130,0.016,0.147,0.996
AVL: minimal value,M-S,220,0.186,0.014,0.204,1.046
AVL: minimal value,M-S (final),220,0.185,0.014,0.204,1.045
AVL: Rebalance,M,4926,1.500,0.000,1.500,2.330
AVL: Rebalance,M-C,4926,1.499,0.061,1.561,2.408
AVL: Rebalance,M-S,4926,1.712,0.052,1.769,2.614
AVL: Left-Right Rotation,M,1534,0.628,0.034,0.663,1.507
AVL: Right Rotation,M,201,0.142,0.000,0.143,0.974
AVL: Search and back to root,M,455,0.140,0.000,0.140,0.971
AVL: Search and back to root,M-C,455,0.140,0.008,0.149,0.999
AVL: Search and back to root,M-S,455,0.144,0.009,0.158,1.006
AVL: Search and swap,M,4104,0.822,0.000,0.823,1.658
AVL: Search and swap,M-C,4104,0.746,0.032,0.779,1.625
AVL: Search and swap,M-S,4104,1.021,0.080,1.106,1.950
DSW Tree Traversal,M,1334,0.252,0.000,0.252,0.998
AVL Tree binary search,M,192,0.086,0.000,0.087,0.925
AVL Tree binary search,M-S,192,0.133,0.014,0.152,0.992
SLL Bubblesort,M,287,0.101,0.000,0.101,0.836
Delete element from SLL,M,152,0.076,0.000,0.076,0.796
Build and insert into DLL,M,379,0.142,0.000,0.142,0.879
DLL insertion sort 1,M,4302,1.134,0.000,1.134,1.876
DLL insertion sort 2,M,1332,0.369,0.000,0.369,1.113
SLL insertsort,M,369,0.114,0.000,0.114,0.848
Build and reverse DLL,M,277,0.122,0.000,0.122,0.859
Construct and reverse SLL,M,95,0.068,0.000,0.068,0.805
Construct and reverse SLL,M-R,95,0.065,0.008,0.083,0.833
Construct and reverse SLL,M-S,95,0.075,0.010,0.086,0.832
Build list of cyclic lists,M,313,0.115,0.000,0.115,0.858
Build and traverse cyclic DLL,M,104,0.079,0.000,0.079,0.819
Build and traverse SLL with head ptr,M,111,0.077,0.000,0.077,0.812
Build and traverse tree,M,44,0.046,0.000,0.046,0.789
Faulty DLL reversal that first traverses the list,M-C-X,93,0.110,0.010,0.121,0.932
Find last element of DLL,M-C-Y,46,0.076,0.011,0.087,0.889
SLL find middle element,M,75,0.052,0.000,0.052,0.772
SLL find middle element,M-N,456,0.110,0.019,0.129,0.895
SLL find middle element,M-R,89,0.066,0.011,0.088,0.823
SLL find middle element,M-R-X,65,0.069,0.010,0.090,0.845
SLL find middle element,M-S,98,0.079,0.011,0.091,0.821
SLL find middle element,M-V,384,0.094,0.019,0.114,0.889
SLL find middle element,M-C-X,421,0.138,0.024,0.162,0.943
AVL Tree insertion,M,10388,2.385,0.000,2.386,3.212
AVL Tree insertion,M-S,10388,2.711,0.365,3.081,3.924
AVL Tree leftmost insertion,M,6120,1.564,0.000,1.564,2.396
AVL Tree leftmost insertion,M-S,6120,1.542,0.041,1.587,2.428
Lindstrom algorithm,M-C,229,0.089,0.008,0.097,0.849
Lindstrom algorithm,M,229,0.084,0.000,0.084,0.813
Lindstrom algorithm,M-N,67941,8.325,3.275,11.600,12.426
Lindstrom algorithm,M-S,229,0.115,0.044,0.160,0.906
Lindstrom algorithm,M-V,2583,0.328,0.065,0.393,1.180
Build tree and traverse recursively,M,1170,0.211,0.000,0.211,0.964
Build tree and traverse recursively,M-S,1170,0.198,0.010,0.209,0.970
Reverse SLL recursively,M,148,0.056,0.000,0.056,0.776
Reverse SLL recursively,M-N-X,2195,0.092,0.014,0.106,0.875
Reverse SLL recursively,M-S,148,0.063,0.009,0.072,0.802
Reverse SLL recursively,M-V-X,1074,0.068,0.011,0.080,0.858
Recursive tree traversal,M,222,0.080,0.000,0.080,0.828
Recursive tree traversal,M-N,295904,8.653,0.046,8.700,9.665
Recursive tree traversal,M-S,972,0.124,0.009,0.133,0.899
DLL reversal,M-C-X,86,0.114,0.016,0.131,0.938
Reverse DLL,M,70,0.054,0.000,0.054,0.823
Reverse DLL,M-N-X,1332,0.700,0.044,0.744,1.673
DLL reversal,M-R,113,0.104,0.008,0.126,0.910
Reverse DLL,M-S,70,0.062,0.020,0.082,0.865
Reverse DLL,M-S-X,35,0.068,0.018,0.086,0.893
Reverse DLL,M-V,572,0.096,0.019,0.115,0.999
Reverse SLL,M,46,0.043,0.000,0.043,0.773
Reverse SLL,M-N-X,383,0.130,0.018,0.147,0.948
Reverse SLL,M-R,76,0.066,0.008,0.084,0.843
Reverse SLL,M-S,46,0.056,0.010,0.066,0.810
Reverse SLL,M-V,170,0.052,0.011,0.064,0.850
AVL Tree binary search and back to root,M-C,455,0.140,0.008,0.149,0.993
AVL Tree binary search and back to root,M,455,0.138,0.000,0.139,0.970
AVL Tree binary search and back to root,M-S,455,0.143,0.009,0.157,1.006
AVL Tree binary search swap,M-C,4104,0.750,0.031,0.782,1.630
AVL Tree search and swap,M,4104,0.816,0.000,0.817,1.641
AVL Tree search and swap,M-S,4104,1.012,0.078,1.095,1.938
Skip list insert,M,330,0.117,0.000,0.117,0.863
Skip list insert,M,302,0.106,0.000,0.106,0.850
AVL Tree transform SLL to AVL Tree,M-C,7166,1.758,0.026,1.784,2.639
AVL Tree transform SLL to AVL Tree,M,7166,1.785,0.000,1.786,2.621
AVL Tree transform SLL to AVL Tree,M-S,7166,1.865,0.022,1.892,2.741
Traverse DLL,M-S,38,0.052,0.015,0.068,0.850
SLL traversal,M,13,0.030,0.000,0.030,0.734
SLL traversal,M-N,97,0.045,0.002,0.047,0.814
SLL traversal,M-R,44,0.049,0.007,0.067,0.797
SLL traversal,M-S,23,0.041,0.009,0.050,0.779
SLL traversal,M-V,47,0.021,0.009,0.031,0.793
Faulty SLL traversal,M-X,19,0.039,0.010,0.049,0.793
Recursive SLL construction,M,103,0.060,0.000,0.060,0.782
Recursive SLL traversal,M,82,0.058,0.000,0.058,0.798
Recursive zip to SLLs,M,575,0.118,0.000,0.118,0.855
Recursive tree traversal,M,1170,0.214,0.000,0.214,0.962
\end{filecontents*}
\csvreader[head to column names]{benchmark-results.csv}{}
        {\name~(\property) & \states & \generation & \mc & \verification & \total \\}
               & & & & & \\\hline
\end{longtable}

\subsection{Programming Language \& Concrete Semantics}\label{app:sec-concrete-semantics}

\subsubsection{Programming Language}

For the sake of concreteness, we present our analysis in terms of a small heap-manipulating programming language.
Note that our implementation actually supports a a richer set of programming language features, such as (potentially recursive) procedure calls, that have been omitted to improve readability. 

Let $x$ be a variable taken from $\Var$ and $f \in \Fields$ be a field. 
Then the syntax of \emph{programs} $\Progs$ ($\pP$), \emph{Boolean expressions} $\BExps$ ($\pB$), and \emph{pointer expressions} $\PExps$ ($\pE$) is defined by the 
context-free grammar in Figure~\ref{fig:syntax}.

\begin{figure}[h]
{
\begin{align*}
\pP ~::=~ &
\Ass{x}{\pE} \,|\,
\Ass{\F{x}{f}}{\pE} \,|\,
\New{x} \,|\,
\Seq{\pP}{\pP} 
\,|\, \Skip
\tag{$\Progs$} 
\\
 \,|\,
 & \Ite{\pB}{\pP}{\pP} \,|\,
\WhileDo{\pB}{\pP}
\\
\pB ~::=~ & 
\Equals{\pE}{\pE} \,|\,
\Con{\pB}{\pB} \,|\,
\Not{\pB}
\tag{$\BExps$}
\\
\qquad\qquad\quad
\pE ~::=~ &
\Null \,|\,
x  \,|\,
\F{x}{f} 
\tag{$\PExps$}
\end{align*}
}
\caption{Syntax of $\Progs$-programs}
\label{fig:syntax}
\end{figure}

The meaning of $\Progs$-programs is straightforward.
For instance, an assignment 
$\Ass{\F{x}{f}}{\Null}$ sets the $f$-field of the record referenced by variable $x$ to the location  $\Null$.
Formally, the semantics of $\Progs$-programs is given by a function 
\[ \PSem{.} : \Progs \to \HC \to \HC \]
that takes a program $P$ and a program state, i.e. an HC $\iH$, and yields an HC capturing the effect of executing $P$ on $\iH$ (if defined).
In Figure~\ref{fig:semantics} the transformer $\PSem{.}$ is defined inductively on the structure of $\Progs$-programs.
For example, the semantics of an assignment $\Ass{\F{x}{f}}{y}$ first determines the nodes $\ESem{x}$ and $\ESem{y}$ attached to variable edges $x$ and $y$, respectively.
After that, existing outgoing edges of $\ESem{x}$ labeled with $f$ are removed.
Finally, a new edge $e$ from $\ESem{x}$ to $\ESem{y}$ with label $f$ is added.
The semantics of the control-flow structures is standard. 

In the following, we formalize the auxiliary functions and graph transformations used in the definition of the concrete semantics.

\begin{figure}[t]
\begin{align*}
& \PSem{\Ass{x}{\pE}} =
\SetVar{x}{\ESem{\pE}}
\qquad
\PSem{\Ass{\F{x}{f}}{\pE}} =
\SetField{x}{f}{\ESem{\pE}}
\\
& \PSem{\New{x}} =
\AddVar{x}
\qquad
\PSem{\Skip} =
\Id{\HC}
\qquad
\PSem{\Seq{\pP_1}{\pP_2}} =
\PSem{\pP_1} \FComp \PSem{\pP_2}
\\
& \PSem{\Ite{\pB}{\pP_1}{\pP_2}} =
\SemIte{\BSem{\pB}}{\PSem{\pP_1}}{\PSem{\pP_2}}
\\
& \PSem{\WhileDo{\pB}{\pP}} =
\Lfp \, Y \,.\, \left( 
Y \mapsto 
\SemIte{\BSem{\pB}}{\PSem{\pP_1}\FComp Y}{\PSem{\Skip}}
\right)
\end{align*}
\caption{
The semantics of $\Progs$-programs. 
$\ESem{\pE}$ is a partial function that evaluates expression $\pE$ for a given HC $\iH$ to a corresponding node in $\iH$.
The semantics of Boolean conditions $\pB$ is given by the partial function $\BSem{\pB} : \HC \to \{\True,\False\}$. 
Further, the auxiliary functions $\SetVar{x}{\ESem{\pE}}$, $\SetField{x}{f}{\ESem{\pE}}$, and $\AddVar{x}$ are HC transformers of type $\HC \to \HC$ that (re)set variable $x$ to node $\ESem{\pE}$, (re)set the edge labeled with $f$ from the node $\ESem{x}$ to node $\ESem{\pE}$, and add a new node and assign $x$ to it, respectively.
As is standard, the semantics of loops is defined as a least fixed-point, denoted by $\Lfp$.
Formal definitions of all auxiliary functions are provided in the corresponding subsection. 
}
\label{fig:semantics}
\end{figure}

\subsubsection{Auxiliary Functions used within Concrete Semantics}

\paragraph{Conditional Function}

Let $B : S \to \{\True,\False\}$ and $f,g : S \to S$ be partial functions. 
Then the conditional function 
\begin{align*}
  \SemIte{B}{f}{g} : S \to S
\end{align*}
is the partial function given by:
\begin{align*}
  \SemIte{B}{f}{g}(s) ~=~
  \begin{cases}
    f(s) & ~\text{if}~ B(s) = \True \\   
    g(s) & ~\text{if}~ B(s) = \False\\   
    \pBot & ~\text{otherwise}
  \end{cases}
\end{align*}

\paragraph{Semantics of Pointer Expressions}

Let $\iH \in \IHC$.
Moreover, let $x \in \Var$ and $f \in \Fields$.
\begin{align*}
  \ESem{\Null}(\iH) ~=~ & \vNull \\
  \ESem{x}(\iH) ~=~ &
  \begin{cases}
    v & ~\text{if}~ \exists e \in \E_{\iH}^{\{x\}} : \Att_{\iH}(e)(1) = v \\
    \pBot & ~\text{otherwise}
  \end{cases}
  \\
  \ESem{\F{x}{f}}(\iH) ~=~ & 
  \begin{cases}
          v & ~\text{if}~ \ESem{x}(\iH) = u \in \V_{\iH} ~\text{and}~ \exists e \in \E_{\iH}^{\{f\}} : \Att_{\iH}(e) = uv \\
    \pBot & ~\text{otherwise}
  \end{cases}
\end{align*}

\paragraph{Semantics of Boolean Expressions}
\begin{align*}
  \BSem{E_1=E_2}(\iH) ~=~ & 
  \begin{cases}
    \True & ~\text{if}~ \ESem{E_1}(\iH) = \ESem{E_2}(\iH) \neq \pBot \\
    \False & ~\text{if}~ \pBot \neq \ESem{E_1}(\iH) \neq \ESem{E_2}(\iH) \neq \pBot \\
    \pBot & ~\text{otherwise}
  \end{cases}
  \\
  \BSem{B_1 \wedge B_2 } ~=~ & 
  \begin{cases}
    \True & ~\text{if}~ \BSem{B_1} = \BSem{B_2} = \True \\
    \False & ~\text{if}~ \BSem{B_1} = \BSem{B_2} = \False \\
    \pBot & ~\text{otherwise}
  \end{cases}
  \\
  \BSem{\neg B}(\iH) ~=~ & 
  \begin{cases}
    \True & ~\text{if}~ \BSem{B} = \False \\
    \False & ~\text{if}~ \BSem{B} = \True \\
    \pBot & ~\text{otherwise}
  \end{cases}
\end{align*}

\subsubsection{Graph Transformations used within Concrete Semantics}\label{app:concrete}

Let $\iH = (\V,\E,\Lab,\Att,\Ind,\Ext) \in \IHC$. 

\paragraph{$\SetVar{.}{.}$}
Let $e \notin \E$ be a fresh edge.
If $\ESem{\pE} = \pBot$ then $\SetVar{x}{\ESem{\pE}}$ is undefined.
Otherwise, if $\ESem{\pE}(\iH) = u \in \V$, we define:
\begin{align*}
\SetVar{x}{\ESem{\pE}} ~=~ & (
\V, \underbrace{(\E \setminus \E^{\{x\}})}_{=\E'} \cup \{e\}, 
(\Proj{\Lab}{\E'}) \cup \{ e \mapsto x \}, \\
& \quad (\Proj{\Att}{\E'}) \cup \{ e \mapsto u \}, (\Proj{\Ind}{\E'}) \cup \{e \mapsto \iBot\}, \Ext)
\end{align*}

\paragraph{$\SetField{.}{.}{.}$}
Let $e \notin \E$ be a fresh edge.
If $\ESem{\pE} = \pBot$ or $\ESem{x} = \pBot$ then $\SetField{x}{f}{\ESem{\pE}}$ is undefined.
Assume that $\ESem{x} = u \in \V$.
Moreover, let $U = \{ e' \in \E ~|~ \Lab(e') = f ~\text{and}~ \Att(e')(1) = u \}$. 
If $U = \emptyset$ then $\SetField{x}{f}{\ESem{\pE}}$ is undefined.
Otherwise, if $\ESem{E}(\iH) = v \in \V$, we define:
\begin{align*}
\SetField{x}{f}{\ESem{\pE}} ~=~ & (
\V, \underbrace{(\E \setminus U)}_{\E'} \cup \{ e \},
(\Proj{\Lab}{\E'}) \cup \{ e \mapsto f \}, \\
& \quad (\Proj{\Att}{\E'}) \cup \{ e \mapsto uv \}, (\Proj{\Ind}{\E'}) \cup \{e \mapsto \iBot\}, \Ext)
\end{align*}

\paragraph{$\AddVar{.}$}
Let $v \notin \V$ be a fresh node and $e \notin \E$ be a fresh edge. Then:
\begin{align*}
\AddVar{x} ~=~ & (
\V \cup \{ v \}, 
\underbrace{(\E \setminus \E^{\{x\}})}_{=\E'}\cup \{e\}, 
(\Proj{\Lab}{\E'}) \cup \{ e \mapsto x \}, \\
& \quad (\Proj{\Att}{\E'}) \cup \{ e \mapsto v \}, (\Proj{\Ind}{\E'}) \cup \{e \mapsto \iBot\}, \Ext)
\end{align*}

\subsection{Definition of Isomorphic Indexed Heap Configurations}
\label{app:excursus}
\begin{definition}
\label{def:isomorphism}
  Let $\iH,\iK$ be IHCs. 
  Then $\iH$ and $\iK$ are \emph{isomorphic}, written $\iH \Iso \iK$, if and only if there exist
  bijective functions $f : \V_{\iH} \to \V_{\iK}$ and $g : \E_{\iH} \to \E_{\iK}$ such that
  \begin{itemize}
    \item for each $e \in \E_{\iH}$, $\Lab_{\iH}(e) = \Lab_{\iK}(g(e))$ and 
      $\Ind_{\iH}(e) = \Ind_{\iK}(g(e))$,  
    \item for each $e \in \E_{\iH}$, $f(\Att_{\iH}(e)) = \Att_{\iK}(g(e))$, and
    \item $f(\Ext_{\iH}) = \Ext_{\iK}$.
  \end{itemize}
\end{definition}

\subsection{Abstract Semantics}\label{app:sec-abstract-semantics}

In this section, we inductively define the abstract semantics of $\Progs$-programs, which have been formally defined in Appendix~\ref{app:sec-concrete-semantics}.

\begin{align*}
\PASem{P} ~=~ & 
\Mater{\iG}{P} \FComp \PSem{P} \FComp \Canon{\iG}{P}
\\
& \text{where}~ P \in \{\Ass{x}{\pE}, \Ass{\F{x}{\pF}}{\pE}, \New{x}, \Skip \}
\\
\PASem{\Seq{\pP_1}{\pP_2}} ~=~ &
\PASem{\pP_1} \FComp \PASem{\pP_2}
\\
\PASem{\Ite{\pB}{\pP_1}{\pP_2}} ~=~ &
\Mater{\iG}{\pB}
\FComp \SemIte{\BSem{\pB}}{\PASem{\pP_1}}{\PASem{\pP_2}}
\\
\PASem{\WhileDo{\pB}{\pP}} ~=~ &
\Lfp \,.\, \left( Y \mapsto 
\SemIte{\BSem{\pB}}{\PASem{\pP}\FComp Y}{\PASem{\Skip}}
\right)
\end{align*}

Here $\Lfp$ denotes the least fixed point operator.
Note that evaluating a guard $\pB$ might require materialization.
Since $\pB$ is formally not a program, we write $\Mater{\iG}{\pB}$ as a shortcut for
\[ \Mater{\iG}{\Ite{\pB}{\Skip}{\Skip}}~,\]
which ensures that $\pB$ can be evaluated without affecting the program state.

%
%

\subsection{Properties of Indexed Graph Grammars (Proof of Theorem~\ref{thm:ig-props})}\label{app:sec-ig-props}

As before, we write $\iH \Iso \iK$ to denote that $\iH$ and $\iK$ are isomorphic (see Appendix~\ref{app:excursus}).
Moreover, let $T = \Types \setminus N$ the set of terminal labels.

Given a natural number $n \in \Nats$, we write $\iH \iDerive{\iG}^{n} \iK$ to denote that $\iG$ derives $\iK$ from $\iH$ in exactly $n$ steps. Formally,
\begin{itemize}
  \item $\iH \iDerive{\iG}^{0} \iH$, and
  \item $\iH \iDerive{\iG}^{n+1} \iK$ iff $\exists \iR : \iH \iDerive{\iG} \iR$ and $\iR \iDerive{\iG}^{n} \iK$.
\end{itemize}
Then $\iDerive{\iG}^{*} ~=~ \bigcup_{n \in \Nats} \iDerive{\iG}^{n}$.

\bigskip
\noindent\textbf{Theorem~\ref{thm:ig-props}(1).} 
Let $\iG$ be an IG and $\iH,\iK$ be indexed heap configurations. 
Then
$\iH \iDerive{\iG}^{*} \iK$ implies $\Lang{\iG}{\iK} \subseteq \Lang{\iG}{\iH}$.

\begin{proof}
We show for all $n \in \Nats$ that $\iH \iDerive{\iG}^{n} \iK$ implies $\Lang{\iG}{\iH} \subseteq \Lang{\iG}{\iH}$.
By induction on $n \in \Nats$.

\noindent\textbf{I.B.}
For $n = 0$ we have
\begin{align*}
& \iH \iDerive{\iG}^{0} \iK \\
~\Leftrightarrow~ & \iH \Iso \iK \tag{Def. $\iDerive{\iG}^{0}$} \\
~\Rightarrow~ & \Lang{\iG}{\iH} = \Lang{\iG}{\iK}.
\end{align*}

\noindent\textbf{I.H.}
Assume for an arbitrary, but fixed $n \in \Nats$ that $\iH \iDerive{\iG}^{n} \iK$ implies $\Lang{\iG}{\iH} \subseteq \Lang{\iG}{\iH}$.

\noindent\textbf{I.S.}
For $n \mapsto n + 1$ we have
\begin{align*}
& \iH \iDerive{\iG}^{n+1} \iK \\
~\Leftrightarrow~ & 
\exists \iR : \iH \iDerive{\iG} \iR ~\text{and}~ \iR \iDerive{\iG}^{n} \iK 
\tag{Def. $\iDerive{\iG}^{n+1}$} \\
~\Rightarrow~ & 
\exists \iR : \iH \iDerive{\iG} \iR ~\text{and}~ \Lang{\iG}{\iK} \subseteq \Lang{\iG}{\iR} 
\tag{I.H.}
\end{align*}
It then suffices to show that $\Lang{\iG}{\iR} \subseteq \Lang{\iG}{\iH}$ holds as
\[
  \Lang{\iG}{\iK} \subseteq \Lang{\iG}{\iR}
  ~\text{and}~ 
  \Lang{\iG}{\iR} \subseteq \Lang{\iG}{\iH} 
  ~\text{implies}~ \Lang{\iG}{\iK} \subseteq \Lang{\iG}{\iH}.
\]
The remaining proof obligation is shown as follows:
\begin{align*}
& \Lang{\iG}{\iR} \\
~=~ & 
\{ F ~|~ \iR \iDerive{\iG}^{*} F ~\text{and}~ \E_{F} = \E_{F}^{T} \} 
\tag{Def.~\ref{def:language}} \\
~=~ & 
\{ F ~|~ \iH \iDerive{\iG} \iR \iDerive{\iG}^{*} F ~\text{and}~ \E_{F} = \E_{F}^{T} \} 
\tag{$\iH \iDerive{\iG} \iR$ possible by assumption} \\
~\subseteq~ & 
\{ F ~|~ \iH \iDerive{\iG}^{*} F ~\text{and}~ \E_{F} = \E_{F}^{T} \} 
\\
~=~ & 
\Lang{\iG}{\iH}
\tag{Def.~\ref{def:language}}
\end{align*}
\qed
\end{proof}

\bigskip
\noindent\textbf{Theorem~\ref{thm:ig-props}(2).} 
Let $\iG$ be an IG and $\iH,\iK$ be indexed heap configurations. 
Then
\[
\Lang{\iG}{\iH} = 
\begin{cases}
\{ \iH \} &~\text{if}~ \neg \exists e \in \E_{\iH} : \Lab_{\iH}(e) \in N \\
\bigcup_{\iH \iDerive{\iG} \iK} \Lang{\iG}{\iK} &~\text{otherwise.}
\end{cases}
\]
\\
\begin{proof}
\noindent\textbf{Case 1:} 
Assume that $\neg \exists e \in \E_{\iH} : \Lab_{\iH}(e) \in N$.
By Def.~\ref{def:derive}, we know that $\iH \niDerive{\iG}$.
Then
\begin{align*}
& \Lang{\iG}{\iH} \\
~=~ &
\{ \iK ~|~ \iH \iDerive{\iG}^{*} \iK ~\text{and}~ \E_{\iK} = \E_{\iK}^{T} \}
\tag{Def.~\ref{def:language}} \\
~=~ & 
\{ \iK ~|~ \iH \iDerive{\iG}^{0} \iK \}
\tag{$\iH \niDerive{\iG}$, $\iDerive{\iG}^{*} = \bigcup_{n \in \Nats} \iDerive{\iG}^{n}$} \\
~=~ &
\{ \iH \}.
\tag{Def. $\iDerive{\iG}^{0}$, $N,T$ disjoint}
\end{align*}

\noindent\textbf{Case 2:} 
Assume that $\exists e \in \E_{\iH} : \Lab_{\iH}(e) \in N$.
By Def.~\ref{def:language}, $\iH \notin \Lang{\iG}{\iH}$.
If $\Lang{\iG}{\iH} = \emptyset$, there is nothing to show.
Thus assume $\Lang{\iG}{\iH} \neq \emptyset$.
By Def.~\ref{def:derive}, this means that there exists a $\iK$ such that $\iH \iDerive{\iG} \iK$.
Then
\begin{align*}
& \Lang{\iG}{\iH} \\
~=~ &
\{ \iR ~|~ \iH \iDerive{\iG}^{*} \iR ~\text{and}~ \E_{\iR} = \E_{\iR}^{T} \}
\tag{Def.~\ref{def:language}} \\
~=~ &
\{ \iR ~|~ \exists n \in \Nats : \iH \iDerive{\iG}^{n+1} \iR ~\text{and}~ \E_{\iR} = \E_{\iR}^{T} \}
\tag{$\iH \notin \Lang{\iG}{\iH}$ and $\iH \iDerive{\iG}^{0} \iR$ iff $\iH \Iso \iR$} \\
~=~ &
\{ \iR ~|~ \exists \iK \, \exists n \in \Nats : \iH \iDerive{\iG} \iK \iDerive{\iG}^{n} \iR ~\text{and}~ \E_{\iR} = \E_{\iR}^{T} \}
\tag{Def. $\iDerive{\iG}^{n+1}$} \\
~=~ &
\bigcup_{\iH \iDerive{\iG} \iK} \{ \iR ~|~ \exists n \in \Nats : \iK \iDerive{\iG}^{n} \iR ~\text{and}~ \E_{\iR} = \E_{\iR}^{T} \} 
\\
~=~ &
\bigcup_{\iH \iDerive{\iG} \iK} \{ \iR ~|~ \iK \iDerive{\iG}^{*} \iR ~\text{and}~ \E_{\iR} = \E_{\iR}^{T} \} 
\tag{Def.~$\iDerive{\iG}^{*}$} \\
~=~ &
\bigcup_{\iH \iDerive{\iG} \iK} \Lang{\iG}{\iK}.
\tag{Def.~\ref{def:language}}
\end{align*}
\qed
\end{proof}

\bigskip
\noindent\textbf{Theorem~\ref{thm:ig-props}(3).} 
Let $\iG$ be an IG and $\iH$ be an indexed heap configuration. 
Then it is decidable whether $\Lang{\iG}{\iH} = \emptyset$.

\begin{proof}
Let $\iG$ be an IG and $\iH$ be an IHC.
We construct an indexed context-free string grammar $\iC$ 
and a string $\rho$ such that
\begin{align*}
\iL{\iG}(\iH) = \emptyset ~\text{if and only if}~ \iL{\iC}(\rho) = \emptyset.
\end{align*}
Since the emptiness problem for indexed context-free string grammars is decidable
(confer Rozenberg, Salomaa: \enquote{Handbook of Formal Languages}, Vol. 2, 1997),
the emptiness problem for IGs is decidable as well.

The grammar $\iC$ is constructed over the same set of terminals $T$, nonterminals $N$ and index symbols $I$ as $\iG$.
Now, let $\iK$ be an IHC with $\E_{\iK} = \{ e_1, \ldots, e_k \}$
We then define the string 
\[ \sigma_{\iK} ~=~ \Lab_{\iK}(e_1)[\Ind_{\iK}(e_1)] \, \Lab_{\iK}(e_2)[\Ind_{\iK}(e_2)] \, \ldots \, \Lab_{\iK}(e_k)[\Ind_{\iK}(e_k)]. \]
Then the grammar $\iC$ is given by the set of rules
\[ \iC ~=~ \{ \Rule{X[\sigma]}{\sigma_{\iK}} ~|~ \Rule{X,\sigma}{\iK} \in \iG \}. \]
Moreover, we set $\rho = \sigma_{\iH}$.
\qed
\end{proof}

\bigskip
\noindent\textbf{Theorem~\ref{thm:ig-props}(4).} 
Let $\iG$ be an IG and $\iH,\iK$ be indexed heap configurations. 
Then
the inverse language $\iLang{\iG}{\iH}$ is non-empty and finite.

\begin{proof}
Recall that 
$\iH \rDerive{\iG} \iK$ holds iff $\iK \iDerive{\iG} \iH$.
Then, since every IG is increasing (see Remark below Definition~\ref{def:derive}), we have
\begin{align*}
0 \leq |\V_{\iK}| + |\E_{\iK}| ~<~ |\V_{\iH}| + |\E_{\iH}|. 
\tag{$\spadesuit$}
\end{align*}
We now show that $\iLang{\iG}{\iH}$ is finite:
\begin{align*}
& \iLang{\iG}{\iH} \\
~=~ & \{ \iK ~|~ \iH \rDerive{\iG}^{*} \iK ~\text{and}~ \iK \nrDerive{\iG} \}
\tag{Def.~\ref{def:language}} \\
~=~ & \{ \iK ~|~ \iK \iDerive{\iG}^{*} \iH ~\text{and}~ \iK \nrDerive{\iG} \}
\tag{Def.~\ref{def:derive}} \\
~=~ & \{ \iK ~|~ \exists n \in \Nats : \iK \iDerive{\iG}^{n} \iH ~\text{and}~ \iK \nrDerive{\iG} \}
\tag{Def. $\iDerive{\iG}^{*}$} \\
~=~ & \{ \iK ~|~ \exists n \leq m : \iK \iDerive{\iG}^{n} \iH ~\text{and}~ \iK \nrDerive{\iG} \}~,
\tag{set $m = |\V_{\iH}| + |\E_{\iH}|$, by $(\spadesuit)$}
\end{align*}
where the last set is finite.

It remains to prove that $\rL{\iC}(\iH)$ is non-empty:
First, note that $\iH \rDerive{\iG}^{0} \iH$ always holds.
Further, by $(\spadesuit)$, after at most $m$ inverse IG derivation steps, i.e. $\bigcup_{0 \leq n \leq m} \rDerive{\iG}^{n}$, no further inverse IG derivation is possible.
Hence, $\rL{\iC}(\iH) \neq \emptyset$.
\qed
\end{proof}

\subsection{Proof of Theorem~\ref{thm:sound} (Soundness)}
\label{app:sec-sound}

\noindent\textbf{Theorem~\ref{thm:sound}.} 
$\forall \pP \in \Progs \,:\, \con{\iG} \FComp \PSem{P} ~\FOrder~ \PASem{P} \FComp \con{\iG}$.

\begin{proof}
By induction on the structure of $\Progs$-programs. 

\noindent\textbf{Base cases.}
Let $P \in \{ \Ass{x}{\pE}, \Ass{\F{x}{\pF}}{\pE}, \New{x}, \Skip \}$.
Then:
\begin{align*}
&
\con{\iG} \FComp \PSem{P}
\\
~=~ &
\iL{\iG} \FComp \PSem{P}
\tag{$\con{\iG} = \iL{\iG}$}\\
~\FOrder~ &
\Mater{\iG}{P} \FComp \PSem{P} \FComp \iL{\iG}
\tag{Def.~\ref{prop:mater}} \\
~\FOrder~ &
\Mater{\iG}{P} \FComp \PSem{P} \FComp \Canon{\iG}{P} \FComp \iL{\iG}
\tag{Def.~\ref{prop:canon}, $\iL{\iG}$ monotone} \\
~=~ &
\PASem{P} \FComp \iL{\iG}
\tag{Appendix~\ref{app:sec-abstract-semantics}}
\\
~=~ &
\PASem{P} \FComp \con{\iG}
\end{align*}

\noindent\textbf{Sequential Composition.}
\begin{align*}
&
\con{\iG} \FComp \PSem{\Seq{P_1}{P_2}}
\\
~=~ &
\iL{\iG} \FComp \PSem{\Seq{P_1}{P_2}}
\tag{$\con{\iG} = \iL{\iG}$}\\
~=~ &
\iL{\iG} \FComp \PSem{P_1} \FComp \PSem{P_2}
\tag{Appendix~\ref{app:sec-concrete-semantics}}
\\
~=~ &
\PASem{P_1} \FComp \iL{\iG} \FComp \PSem{P_2}
\tag{I.H.} \\
~=~ &
\PASem{P_1} \FComp \PASem{P_2} \FComp \iL{\iG}
\tag{I.H.} \\
~=~ &
\PASem{\Seq{P_1}{P_2}} \FComp \iL{\iG}
\tag{Appendix~\ref{app:sec-abstract-semantics}} 
\\
~=~ &
\PASem{\Seq{P_1}{P_2}} \FComp \con{\iG}
\end{align*}

\noindent\textbf{Conditionals.}

Let $Q = \Ite{\pB}{\Skip}{\Skip}$.
By Fig.~\ref{fig:semantics}, we have $\PSem{Q} = \Id{\HC}$.
Moreover, we have 
\begin{align*}
\forall \iH \in \Mater{\iG}{Q} \, \forall \iK \in \iL{\iG}(\iH) : \BSem{\pB}(\iH) = \BSem{\pB}(\iK). \tag{$\spadesuit$}
\end{align*} 
as, by Def.~\ref{prop:mater}, $\Mater{\iG}{Q} \neq \pBot$.
Thus $\pB$ is evaluated in the same way on both graphs.
Then:
\begin{align*}
&
\con{\iG} \FComp \PSem{\Ite{\pB}{P_1}{P_2}}
\\
~=~ &
\iL{\iG} \FComp \PSem{\Ite{\pB}{P_1}{P_2}}
\tag{$\con{\iG} = \iL{\iG}$} \\
~=~ &
\iL{\iG} \FComp \Id{\HC} \FComp \PSem{\Ite{\pB}{P_1}{P_2}}
\\
~=~ &
\iL{\iG} \FComp \PSem{Q} \FComp \PSem{\Ite{\pB}{P_1}{P_2}}
\tag{$\PSem{Q} = \Id{\HC}$}\\
~\FOrder~ &
\Mater{\iG}{Q} \FComp \PSem{Q} \FComp \iL{\iG} \FComp \PSem{\Ite{\pB}{P_1}{P_2}}
\tag{Def.~\ref{prop:mater}} \\
~=~ & 
\Mater{\iG}{Q} \FComp \PSem{Q} \FComp \iL{\iG} \FComp \SemIte{\BSem{\pB}}{\PSem{P_1}}{\PSem{P_2}}
\tag{Fig.~\ref{fig:semantics}} \\
~=~ & 
\Mater{\iG}{Q} \FComp \iL{\iG} \FComp \SemIte{\BSem{\pB}}{\PSem{P_1}}{\PSem{P_2}}
\tag{$\PSem{Q} = \Id{\HC}$}\\
~=~ &
\Mater{\iG}{Q} \FComp \SemIte{\BSem{\pB}}{(\iL{\iG} \FComp \PSem{P_1})}{(\iL{\iG} \FComp \PSem{P_2})}
\tag{by $\spadesuit$} \\
~\FOrder~ &
\Mater{\iG}{Q} \FComp \SemIte{\BSem{\pB}}{(\PASem{P_1} \FComp \iL{\iG})}{(\PASem{P_2} \FComp \iL{\iG})}
\tag{I.H.} \\
~=~ &
\Mater{\iG}{Q} \FComp \SemIte{\BSem{\pB}}{\PASem{P_1}}{\PASem{P_2}} \FComp \iL{\iG}
\\
~=~ &
\PASem{\Ite{\pB}{P_1}{P_2}} \FComp \iL{\iG}
\tag{Appendix~\ref{app:sec-abstract-semantics}} 
\\
~=~ &
\PASem{\Ite{\pB}{P_1}{P_2}} \FComp \con{\iG}
\end{align*}

\noindent\textbf{Loops.}

We use a standard characterization of the semantics of loops as the supremum of its finite unrollings.
Thus, let 
\begin{align*}
  \WhileKDo{\pB}{\pP}{0} ~=~ & \pBot \\
  \WhileKDo{\pB}{\pP}{k+1} ~=~ & \Ite{\pB}{\Seq{\pP}{\WhileKDo{\pB}{\pP}{k}}}{\Skip}
\end{align*}
be the finite unrollings, where $0$ unrollings corresponds to a program that is undefined everywhere, such as $\WhileDo{\True}{\Skip}$.
Then, by standard arguments, we have
\begin{align*}
        \PSem{\WhileDo{\pB}{\pP}} ~=~ & \sup_{k \in \Nats} \PSem{\WhileKDo{\pB}{\pP}{k}},~\text{and} \\
        \PASem{\WhileDo{\pB}{\pP}} ~=~ & \sup_{k \in \Nats} \PASem{\WhileKDo{\pB}{\pP}{k}}.
\end{align*}
To complete the proof we show for all $k \in \Nats$ that
\begin{align*}
        \con{\iG} \FComp \PSem{\WhileKDo{\pB}{\pP}{k}} ~\FOrder~ \PASem{\WhileKDo{\pB}{\pP}{k}} \FComp \con{\iG}.
\end{align*}

\noindent\textbf{I.B.}
For $k = 0$, we have
\begin{align*}
&
\con{\iG} \FComp \PSem{\WhileKDo{\pB}{\pP}{0}} 
\\
~=~ &
\iL{\iG} \FComp \PSem{\WhileKDo{\pB}{\pP}{0}} 
\tag{$\con{\iG} = \iL{\iG}$} \\
~=~ & 
\iL{\iG} \FComp \pBot
\\
~=~ &
\pBot 
\\
~=~ &
\pBot \FComp \iL{\iG}
\\
~=~ &
\pBot \FComp \con{\iG}
\end{align*}

\noindent\textbf{I.S.}
For $k \mapsto k+1$, we have
\begin{align*}
&
\con{\iG} \FComp \PSem{\WhileKDo{\pB}{\pP}{k+1}} 
\\
~=~ &
\iL{\iG} \FComp \PSem{\WhileKDo{\pB}{\pP}{k+1}} 
\tag{$\con{\iG} = \iL{\iG}$} \\
~=~ &
\iL{\iG} \FComp \PSem{\Ite{\pB}{\Seq{\pP}{\WhileKDo{\pB}{\pP}{k}}}{\Skip}}
\tag{Def. $\WhileKDo{.}{.}{k+1}$} \\
~\FOrder~ & \Mater{\iG}{Q} \\
& \FComp \SemIte{\BSem{\pB}}{(\iL{\iG} \FComp \PSem{\Seq{\pP}{\WhileKDo{\pB}{\pP}{k}}})}{(\iL{\iG} \FComp \PSem{\Skip})}
\tag{analogously to the case of conditionals, where $Q = \Ite{\pB}{\Skip}{\Skip}$} \\
~=~ & \Mater{\iG}{Q} \\
          & \FComp \SemIte{\BSem{\pB}}{(\iL{\iG} \FComp \PSem{\pP} \FComp \PSem{\WhileKDo{\pB}{\pP}{k}})}{(\iL{\iG} \FComp \PSem{\Skip})}
\tag{Fig.~\ref{fig:semantics}} \\
~\FOrder~ & \Mater{\iG}{Q} \\
          & \FComp \SemIte{\BSem{\pB}}{(\PASem{\pP} \FComp \iL{\iG} \FComp \PSem{\WhileKDo{\pB}{\pP}{k}})}{(\PASem{\Skip} \FComp \iL{\iG})}
\tag{outer I.H. on $\Skip$, $\pP$} \\
~\FOrder~ & \Mater{\iG}{Q} \\
          & \FComp \SemIte{\BSem{\pB}}{(\PASem{\pP} \FComp \PASem{\WhileKDo{\pB}{\pP}{k}} \FComp \iL{\iG})}{(\PASem{\Skip} \FComp \iL{\iG})}
\tag{inner I.H. on $\WhileKDo{\pB}{\pP}{k}$} \\
~=~ & \Mater{\iG}{Q} \\
    & \FComp \SemIte{\BSem{\pB}}{\PASem{\Seq{\pP}{\WhileKDo{\pB}{\pP}{k}}}}{\PASem{\Skip}} \FComp \iL{\iG}
\\
~=~ &
\PASem{\Ite{\pB}{\Seq{\pP}{\WhileKDo{\pB}{\pP}{k}}}{\Skip}} \FComp \iL{\iG}
\tag{Semantics of $\Ite{.}{.}{.}$} \\
~=~ & 
\PASem{\WhileKDo{\pB}{\pP}{k+1}} \FComp \iL{\iG}
\tag{Def. $\WhileKDo{\pB}{\pP}{k+1}$} \\
~=~ & 
\PASem{\WhileKDo{\pB}{\pP}{k+1}} \FComp \con{\iG}
\end{align*}
\qed
\end{proof}

\subsection{Local Reasoning}
\label{app:sec-local}

Our analysis enjoys a local reasoning property that is similar to the frame rule in separation logic.
In order to formulate this property, we write $\Mod{\pP}{\iH}$ to denote the set of all nodes and edges that are added or deleted when running $\pP$ on $\iH$.

More precisely, the function $\Mod{.}{.}$ is defined inductively on the structure of $\Progs$ programs (see Appendix~\ref{app:sec-concrete-semantics}).
For the base cases, let $\iH \in \IHC$ and $\iK = \PSem{\pP}(\iH)$.
Then, $\Mod{\pP}{\iH}$ is given by:
\begin{align*}
        \Mod{\Ass{x}{\pE}}{\iH} ~=~ & (\E_{\iH} \setminus \E_{\iK}) \cup (\E_{\iK} \setminus \E_{\iH}) \\ 
        \Mod{\Ass{\F{x}{f}}{\pE}}{\iH} ~=~ & (\E_{\iH} \setminus \E_{\iK}) \cup (\E_{\iK} \setminus \E_{\iH}) \\ 
        \Mod{\New{x}}{\iH} ~=~ & (\V_{\iK} \setminus \V_{\iH}) \cup (\E_{\iH} \setminus \E_{\iK}) \cup (\E_{\iK} \setminus \E_{\iH}) \\ 
        \Mod{\Skip}{\iH} ~=~ & \emptyset
\end{align*}

The composite cases are defined as follows:
\begin{align*}
        \Mod{\Seq{\pP_1}{\pP_2}}{\iH} ~=~ & \Mod{\pP_1}{\iH} \cup \Mod{\pP_2}{\PSem{\pP}{\iH}} \\
        \Mod{\Ite{\pB}{\pP_1}{\pP_2}}{\iH} ~=~ & \Mod{\pP_1}{\iH} \cup \Mod{\pP_2}{\iH} \\
        \Mod{\WhileDo{\pB}{\pP}}{\iH} ~=~ & \bigcup_{k \in \Nats} \Mod{\WhileKDo{\pB}{\pP}{k}}{\iH}~,
\end{align*}
where $\WhileKDo{\pB}{\pP}{k}$ denotes the $k$-th loop unrolling of the loop $\WhileDo{\pB}{\pP}$.
Formally,
\begin{align*}
    \WhileKDo{\pB}{\pP'}{0} ~=~ & \pBot \\
    \WhileKDo{\pB}{\pP'}{k+1} ~=~ & \Ite{\pB}{\Seq{\pP'}{\WhileKDo{\pB}{\pP'}{k}}}{\Skip}~,
\end{align*}
where $0$ unrollings corresponds to a program that is undefined everywhere, such as $\WhileDo{\True}{\Skip}$.

Moreover, $\iH \cup \iR$ denotes the componentwise union of two (not necessarily disjoint) IHCs $\iH$ and $\iR$.
To be precise, we define $\iH \cup \iR = (V_{\iH} \cup \V_{\iR}, \E_{\iH} \cup \E_{\iR}, \Lab_{\iH} \cup \Lab_{\iR}, \Att_{\iH} \cup \Att_{\iR}, \Ind_{\iH} \cup \Ind_{\iR}, \Ext_{\iH})$. Moreover, if $\Ext_{\iH} \neq \Ext_{\iR}$ or $(\iH \cup \iR) \notin \IHC$, we set $\iH \cup \iR$ to be undefined.

\begin{theorem}[Local Reasoning]\label{thm:local}
    Let $\pP \in \Progs$ and $\iH,\iR$ be IHCs. Then
    $\PSem{\pP}(\iH) = \iK$ and $\Mod{\pP}{\iH} \cap (\V_{\iR} \cup \E_{\iR}) = \emptyset$ implies $\PSem{\pP}(\iH \cup \iR) = \iK \cup \iR$.
\end{theorem}

\begin{proof}
        Without loss of generality, we assume $\iH \cup \iR \in \IHC$ and $\Ext_{\iH} = \Ext_{\iR}$. 
        Otherwise, $\iH \cup \iR$ is undefined.
        Then it is straightforward to show that also $\iK \cup \iR$ and $\PSem{\pP}{\iH \cup \iR}$ are undefined, i.e., the theorem holds.

        Let $\PSem{\pP}(\iH) = \iK$ and $\Mod{\pP}{\iH} \cap (\V_{\iR} \cup \E_{\iR}) = \emptyset$.
        The proof proceeds by induction on the structure of $\Progs$ programs.

        \noindent\textbf{I.B.}

        The case $\pP = (\Ass{x}{\pE})$.
        \begin{align*}
                & \PSem{\pP}(\iH \cup \iR) \\
                ~=~ & \SetVar{x}{\ESem{\pE}}(\iH \cup \iR) \tag{Fig.~\ref{fig:semantics}}\\
                ~=~ & (
                        \V_{\iH} \cup \V_{\iR},
                        \underbrace{(\E_{\iH} \cup \E_{\iR}) \setminus \E_{\iH}^{\{x\}}}_{\E'},
                        \Proj{(\Lab_{\iH} \cup \Lab_{\iR})}{\E'},
                    \tag{App.~\ref{app:concrete}} \\
                    & \quad \Proj{(\Att_{\iH} \cup \Att_{\iR})}{\E'},
                        \Proj{(\Ind_{\iH} \cup \Ind_{\iR})}{\E'},
                        \Ext_{\iH}
                      ) \\
                ~=~ & (
                        \V_{\iH} \cup \V_{\iR},
                        \underbrace{(\E_{\iH} \setminus \E_{\iH}^{\{x\}}) \cup \E_{\iR}}_{\E'},
                        \Proj{(\Lab_{\iH} \cup \Lab_{\iR})}{\E'}, \tag{$\E_{\iH}^{\{x\}} \in \Mod{\pP}{\iH}$} \\
                    & \quad \Proj{(\Att_{\iH} \cup \Att_{\iR})}{\E'},
                        \Proj{(\Ind_{\iH} \cup \Ind_{\iR})}{\E'},
                        \Ext_{\iH}
                      ) \\
                %
                ~=~ & (
                        \V_{\iH} \cup \V_{\iR},
                        \underbrace{(\E_{\iH} \setminus \E_{\iH}^{\{x\}})}_{\E''} \cup \E_{\iR},
                        (\Proj{\Lab_{\iH}}{E''}) \cup \Lab_{\iR}, \\
                    & \quad (\Proj{\Att_{\iH}}{\E''}) \cup \Att_{\iR},
                        (\Proj{\Ind_{\iH}}{\E''}) \cup \Ind_{\iR},
                        \Ext_{\iH}
                    )
                \tag{$E' = \E'' \cup \E_{\iR}$} \\
                ~=~ & \SetVar{x}{\ESem{\pE}}(\iH) \cup \iR 
                \tag{$\iH \cup \iR$ is the componentwise union} \\
                ~=~ & \iK. \tag{Fig.~\ref{fig:semantics}, assumption}
        \end{align*}

        The case $\pP = (\Ass{\F{x}{f}}{\pE})$.
        \begin{align*}
                & \PSem{\pP}(\iH \cup \iR) \\
                ~=~ & \SetField{x}{f}{\ESem{\pE}}(\iH \cup \iR) \tag{Fig.~\ref{fig:semantics}} \\
                ~=~ & (
                    \V_{\iH} \cup \V_{\iR}, 
                    \underbrace{((\E_{\iH} \cup \E_{\iR}) \setminus U )}_{\E'} \cup \{e\}, \\
                    & \quad \Proj{(\Lab_{\iH} \cup \Lab_{\iR})}{\E'} \cup \{ e \mapsto f \},
                      \Proj{(\Att_{\iH} \cup \Att_{\iR})}{\E'} \cup \{ e \mapsto uv \}, \\
                    & \quad \Proj{(\Ind_{\iH} \cup \Ind_{\iR})}{\E'} \cup \{ e \mapsto z \}, \Ext_{\iH}
                    )
            \tag{App.~\ref{app:concrete}, where $U,e,u,v$ are as in the Definition of $\SetField{x}{f}{\ESem{\pE}}$} \\
                ~=~ & (
                    \V_{\iH} \cup \V_{\iR}, 
                      \underbrace{(\E_{\iH} \setminus U )}_{\E''} \cup \{e\}) \cup \E_{\iR},
                      \tag{$U,e \in \Mod{\pP}{\iH}$} \\
                    & \quad \Proj{\Lab_{\iH}}{\E''} \cup \{ e \mapsto f \} \cup \Lab_{\iR},
                      \Proj{\Att_{\iH}}{\E''} \cup \{ e \mapsto uv \} \cup \Att_{\iR}, \\
                    & \quad \Proj{\Ind_{\iH} }{\E''} \cup \{ e \mapsto z \} \cup \Ind_{\iR}, \Ext_{\iH}
                    ) \\
                ~=~ & \SetField{x}{f}{\ESem{\pE}}(\iH) \cup \iR \tag{App.~\ref{app:concrete}} \\
                ~=~ & \iK \cup \iR. \tag{Fig~\ref{fig:semantics}, assumption}
        \end{align*}

        The case $\pP = (\New{x})$. 
        \begin{align*}
                & \PSem{\pP}(\iH \cup \iR) \\
                ~=~ & \AddVar{x}(\iH \cup \iR) \tag{Fig.~\ref{fig:semantics}}\\
                ~=~ & (
                        \V_{\iH} \cup \V_{\iR} \cup \{v\},
                        \underbrace{(\E_{\iH} \cup \E_{\iR}) \setminus \E_{\iH}^{\{x\}} \cup \{e\}}_{\E'},
                    \tag{App.~\ref{app:concrete}, v,e~\text{fresh}} \\
                    & \quad \Proj{(\Lab_{\iH} \cup \Lab_{\iR})}{\E'} \cup \{ e \mapsto x \}, 
                        \Proj{(\Att_{\iH} \cup \Att_{\iR})}{\E'} \cup \{ e \mapsto v \}, \\
                    & \quad \Proj{(\Ind_{\iH} \cup \Ind_{\iR})}{\E'} \cup \{ e \mapsto z\},
                        \Ext_{\iH}
                      ) \\ 
                ~=~ & (
                        \V_{\iH} \cup \V_{\iR} \cup \{v\},
                        \underbrace{(\E_{\iH} \setminus \E_{\iH}^{\{x\}})}_{\E''} \cup \{e\} \cup \E_{\iR}, 
                        \tag{$e, \E_{\iH}^{\{x\}} \in \Mod{\pP}{\iH}$} \\
                        & \quad (\Proj{\Lab_{\iH}}{\E''} \cup \{ e \mapsto x \}) \cup \Lab_{\iR}, 
                        (\Proj{\Att_{\iH}}{\E''} \cup \{ e \mapsto v \}) \cup \Att_{\iR}, \\
                        & \quad \Proj{\Ind_{\iH}}{\E''} \cup \{ e \mapsto z\} \cup \Ind_{\iR},
                        \Ext_{\iH}
                      ) \\
               ~=~ & \AddVar{x}(\iH) \cup \iR \tag{App.~\ref{app:concrete}} \\
               ~=~ & \iK \cup \iR. \tag{Fig.~\ref{fig:semantics}, assumption}
        \end{align*}

        The case $\pP = (\Skip)$. Then
        \begin{align*}
                & \PSem{\pP}(\iH \cup \iR) \\
            ~=~ & \iH \cup \iR \tag{Fig.~\ref{fig:semantics}} \\
            ~=~ & \iK \cup \iR. \tag{$\iK = \PSem{\pP}(\iH)$ by assumption}
        \end{align*}

        \noindent\textbf{I.H.}
        Assume for all (sub-)programs $\pP$ and all $\iH,\iR \in \IHC$ that
        \[ \PSem{\pP}(\iH) = \iK ~\text{and}~ \Mod{\pP}{\iH} \cap (\V_{\iR} \cup \E_{\iR}) = \emptyset ~\text{implies}~ \PSem{\pP}(\iH \cup \iR) = \iK \cup \iR. \]
        
        \noindent\textbf{I.S.}

        The case $\pP = (\Seq{\pP_1}{\pP_2})$.
        \begin{align*}
                & \PSem{\pP}(\iH \cup \iR) \\
                ~=~ & (\PSem{\pP_1} \FComp \PSem{\pP_2})(\iH \cup \iR) \tag{Fig.~\ref{fig:semantics}} \\
                ~=~ & \PSem{\pP_2}(\PSem{\pP_1}(\iH \cup \iR)) \tag{Definition of $\FComp$} \\
                ~=~ & \PSem{\pP_2}(\PSem{\pP_1}(\iH) \cup \iR) \tag{$\Mod{\pP_1}{\iH} \subseteq \Mod{\pP}{\iH}$, I.H.} \\
                ~=~ & \PSem{\pP_2}(\PSem{\pP_1}(\iH)) \cup \iR \tag{$\Mod{\pP_2}{\PSem{\pP_1}(\iH)} \subseteq \Mod{\pP}{\iH}$, I.H.} \\ 
                ~=~ & \PSem{\pP}(\iH) \cup \iR. \tag{Fig~\ref{fig:semantics}}
        \end{align*}

        The case $\pP = (\Ite{\pB}{\pP_1}{\pP_2})$.
        \begin{align*}
                & \PSem{\pP}(\iH \cup \iR) \\
                ~=~ & (\SemIte{\BSem{\pB}}{\PSem{\pP_1}}{\PSem{\pP_2}})(\iH \cup \iR) \tag{Fig.~\ref{fig:semantics}} \\
                ~=~ & \SemIte{\BSem{\pB}(\iH)}{\PSem{\pP_1}(\iH \cup \iR)}{\PSem{\pP_2}(\iH \cup \iR)} 
                    \tag{by assumption $\BSem{\pB}(\iH) = \BSem{\pB}(\iH \cup \iR)$} \\
                ~=~ & \SemIte{\BSem{\pB}(\iH)}{\PSem{\pP_1}(\iH) \cup \iR}{\PSem{\pP_2}(\iH) \cup \iR} \tag{I.H.} \\
                ~=~ & \SemIte{\BSem{\pB}(\iH)}{\PSem{\pP_1}(\iH)}{\PSem{\pP_2}(\iH)} \cup \iR \tag{Algebra} \\
                ~=~ & (\SemIte{\BSem{\pB}}{\PSem{\pP_1}}{\PSem{\pP_2}})(\iH) \cup \iR \\
                ~=~ & \iK \cup \iR. \tag{Fig.~\ref{fig:semantics}}
        \end{align*}

        The case $\pP = (\WhileDo{\pB}{\pP'})$. As in proof of Theorem~\ref{thm:sound}, we use a standard characterization of the semantics of loops as the supremum of its finite unrollings.
        Thus, let 
        \begin{align*}
          \WhileKDo{\pB}{\pP'}{0} ~=~ & \pBot \\
          \WhileKDo{\pB}{\pP'}{k+1} ~=~ & \Ite{\pB}{\Seq{\pP'}{\WhileKDo{\pB}{\pP'}{k}}}{\Skip}
        \end{align*}
        be the finite unrollings, where $0$ unrollings corresponds to a program that is undefined everywhere, such as $\WhileDo{\True}{\Skip}$.
        Then, by standard arguments, we have
        \begin{align*}
            \PSem{\WhileDo{\pB}{\pP'}} ~=~ & \sup_{k \in \Nats} \PSem{\WhileKDo{\pB}{\pP'}{k}}
        \end{align*}
        To complete the proof we show for all $k \in \Nats$ that
        \begin{align*}
            & \PSem{\WhileKDo{\pB}{\pP'}{k}}(\iH) 
            ~\text{and}~ \Mod{\pP}{\iH} \cap (\V_{\iR} \cup \E_{\iR}) = \emptyset \\
            & \qquad ~\text{implies}~ \PSem{\WhileKDo{\pB}{\pP'}{k}}(\iH \cup \iR) = \PSem{\WhileKDo{\pB}{\pP'}{k}}(\iH) \cup \iR.
        \end{align*}
        The proof proceeds by induction on $k$.
        The base cases $k=0$ and $k=1$ are trivial.
        For the induction step, we have
        \begin{align*}
            & \PSem{\WhileKDo{\pB}{\pP'}{k+1}}(\iH \cup \iR) \\
        ~=~ & \PSem{\Ite{\pB}{\Seq{\pP}{\WhileKDo{\pB}{\pP'}{k}}}{\Skip}}(\iH \cup \iR) \tag{Def. of $\WhileKDo{\pB}{\pP'}{k+1}$} \\
        ~=~ & (\SemIte{\BSem{\pB}}{\PSem{\Seq{\pP_1}{\WhileKDo{\pB}{\pP'}{k}}}}{\PSem{\Skip}})(\iH \cup \iR) \\
              \tag{Fig.~\ref{fig:semantics}} \\
        ~=~ & \SemIte{\BSem{\pB}(\iH)}{\PSem{\Seq{\pP_1}{\WhileKDo{\pB}{\pP'}{k}}}(\iH \cup \iR)}{\PSem{\Skip}(\iH \cup \iR)} 
              \tag{by assumption $\BSem{\pB}(\iH) = \BSem{\pB}(\iH \cup \iR)$} \\
        ~=~ & \SemIte{\BSem{\pB}(\iH)}{\PSem{\WhileKDo{\pB}{\pP'}{k}}(\PSem{\pP'}(\iH \cup \iR))}{\PSem{\Skip}(\iH \cup \iR)} 
              \tag{Fig.~\ref{fig:semantics}} \\
      ~=~ & \SemIte{\BSem{\pB}(\iH)}{\PSem{\WhileKDo{\pB}{\pP'}{k}}(\PSem{\pP'}(\iH) \cup \iR)}{\PSem{\Skip}(\iH) \cup \iR}
              \tag{outer I.H. on $\pP'$ and $\Skip$} \\
      ~=~ & \SemIte{\BSem{\pB}(\iH)}{\PSem{\WhileKDo{\pB}{\pP'}{k}}(\PSem{\pP'}(\iH)) \cup \iR}{\PSem{\Skip}(\iH) \cup \iR}
              \tag{outer I.H. on $\WhileKDo{\pB}{\pP'}{k}$} \\
      ~=~ & \SemIte{\BSem{\pB}(\iH)}{\PSem{\WhileKDo{\pB}{\pP'}{k}}(\PSem{\pP'}(\iH))}{\PSem{\Skip}(\iH)} \cup \iR \\
      ~=~ & (\SemIte{\BSem{\pB}}{\PSem{\WhileKDo{\pB}{\pP'}{k}}(\PSem{\pP'})}{\PSem{\Skip}}(\iH) \cup \iR \\
      ~=~ & \PSem{\WhileKDo{\pB}{\pP'}{k+1}}(\iH) \cup \iR. \tag{Def. $\WhileKDo{\pB}{\pP'}{k+1}$}
        \end{align*}
\qed
\end{proof}

\subsection{Expressiveness of Backward Confluent Grammars}\label{app:sec-backwardconfluent-expressiveness}

\begin{theorem}
  There exist languages of IHCs that can be generated by an IG, but not by any backward confluent IG.
\end{theorem}

\begin{proof}
Given a nonterminal symbol, say $S$, let $e_{S}$ denote a single hyperedge that is labeled with $S$ and attached to $\Rank(S)$ distinct nodes.
Now, consider the IG $G$ below 
in which all indices as well as the $\Null$ node are omitted for simplicity. 
\begin{center}
\scalebox{0.7}{
\tikzdefault{
   
    \node (text) {$\texttt{S} ~\rightarrow~$};

    \begin{scope}[shift={(1,0)}]
        \location{v1}{}{1}
        \location{v2}{right of = v1}{2}
        \lpointer{v1}{}{}{v2}{a}{above}
    
        \node[right of = v1, node distance=2cm] (s1) {};
        \node[above of = s1, node distance=0.5cm] (s0) {};
        \node[below of = s1, node distance=0.5cm] (s2) {};
        \path (s0) edge[-] (s2);

    \end{scope}

    \begin{scope}[shift={(4.5,0)}]
        \location{v1}{}{1}
        \location{v2}{right of = v1}{2}
        \lpointer{v1}{}{}{v2}{b}{above}
        
        \node[right of = v1, node distance=2cm] (s1) {};
        \node[above of = s1, node distance=0.5cm] (s0) {};
        \node[below of = s1, node distance=0.5cm] (s2) {};
        \path (s0) edge[-] (s2);
    \end{scope}
    
    \begin{scope}[shift={(8,0)}]
        \location{v1}{}{1}
        \location{v2}{right of = v1}{}
        \lpointer{v1}{}{}{v2}{a}{above}
        \nonterminal{S}{right}{v2}{S}
        \location{v3}{right of = S}{2}
        \tentacle{S}{1}{v2}{}
        \tentacle{S}{2}{v3}{}
    \end{scope}
}
}

\end{center}

$G$ generates non-empty singly-linked list segments whose next pointer is either labeled with $a$ or $b$.
More precisely, $\iL{G}(e_{S})$ consists of all IHCs that are non-empty singly-linked list segments that either 
start with arbitrarily many (including zero) fields labeled $a$ and then end with exactly one field labeled $b$, or
consist only of fields labeled with $a$.
We claim that $\iL{G}(e_S)$ cannot be generated by a backward confluent IG.
Towards a contradiction, assume there exists a backward confluent IG $G'$ and a nonterminal $S'$ such that $\iL{G'}(e_{S'}) = \iL{G}(e_{S})$.

Clearly, each of the following three IHCs (on the left-hand side of $\rDerive{\iG}^{*}$) belongs to $\iL{G}(e_S)$.
Hence, the following three inverse derivations are possible:
\begin{center}
\scalebox{0.7}{
\tikzdefault{
   
    \begin{scope}[shift={(0,0)}]
        \location{v1}{}{}
        \location{v2}{right of = v1}{}
        \location{v3}{right of = v2}{}
        \lpointer{v1}{}{}{v2}{a}{above}
        \lpointer{v2}{}{}{v3}{b}{above}
        
        \node[right of = v3, node distance = 0.5cm] (text) {$,$};
        
    \end{scope}
    
    \begin{scope}[shift={(3,0)}]
        \location{v1}{}{}
        \location{v2}{right of = v1}{}
        \location{v3}{right of = v2}{}
        \lpointer{v1}{}{}{v2}{a}{above}
        \lpointer{v2}{}{}{v3}{a}{above}
        
        \node[right of = v3, node distance = 0.5cm] (text) {$,$};
        
    \end{scope}
    
    \begin{scope}[shift={(6,0)}]
        \location{v1}{}{}
        \location{v2}{right of = v1}{}
        \lpointer{v1}{}{}{v2}{a}{above}
        
        \node[right of = v2, node distance = 1.5cm] (text) {$\rDerive{G'}^{*}$};
    \end{scope}

    \begin{scope}[shift={(10,0)}]
        \location{v1}{}{}
        \nonterminal{S1}{right}{v1}{S'}
        \location{v2}{right of = S1}{}

        \tentacle{S1}{1}{v1}{}
        \tentacle{S1}{2}{v2}{}
    \end{scope}

}
}

\end{center}

Since $G'$ is backward confluent, this means that also the following inverse derivation is possible:

\begin{center}
\scalebox{0.7}{
\tikzdefault{
   
    \begin{scope}[shift={(0,0)}]
        \location{v1}{}{}
        \location{v2}{right of = v1}{}
        \location{v3}{right of = v2}{}
        \lpointer{v1}{}{}{v2}{a}{above}
        \lpointer{v2}{}{}{v3}{a}{above}
        
        \node[right of = v3, node distance = 1cm] (text) {$\rDerive{G'}^{*}$};
    \end{scope}

    \begin{scope}[shift={(4,0)}]
        \location{v1}{}{}
        \nonterminal{S1}{right}{v1}{S'}
        \location{v2}{right of = S1}{}
        \nonterminal{S2}{right}{v2}{S'}
        \location{v3}{right of = S2}{}

        \tentacle{S1}{1}{v1}{}
        \tentacle{S1}{2}{v2}{}
        
        \tentacle{S2}{1}{v2}{}
        \tentacle{S2}{2}{v3}{}
        
        \node[right of = v3, node distance = 1cm] (text) {$\rDerive{G'}^{*}$};
    \end{scope}

    \begin{scope}[shift={(10,0)}]
        \location{v1}{}{}
        \nonterminal{S1}{right}{v1}{S'}
        \location{v2}{right of = S1}{}

        \tentacle{S1}{1}{v1}{}
        \tentacle{S1}{2}{v2}{}
    \end{scope}

}
}

\end{center}

Then we can also apply an 
inverse derivation that starts in an IHC $H$:

\begin{center}
\scalebox{0.7}{
\tikzdefault{
   
    \begin{scope}[shift={(0,0)}]
        \location{v1}{}{}
        \location{v2}{right of = v1}{}
        \location{v3}{right of = v2}{}
        \location{v4}{right of = v3}{}
        \location{v5}{right of = v4}{}
        \lpointer{v1}{}{}{v2}{a}{above}
        \lpointer{v2}{}{}{v3}{b}{above}
        \lpointer{v3}{}{}{v4}{a}{above}
        \lpointer{v4}{}{}{v5}{b}{above}
        
        \node[right of = v5, node distance = 0.8cm] (text) {$\rDerive{G'}^{*}$};
    \end{scope}

    \begin{scope}[shift={(5.5,0)}]
        \location{v1}{}{}
        \nonterminal{S1}{right}{v1}{S'}
        \location{v2}{right of = S1}{}
        \nonterminal{S2}{right}{v2}{S'}
        \location{v3}{right of = S2}{}

        \tentacle{S1}{1}{v1}{}
        \tentacle{S1}{2}{v2}{}
        
        \tentacle{S2}{1}{v2}{}
        \tentacle{S2}{2}{v3}{}
        
        \node[right of = v3, node distance = 0.8cm] (text) {$\rDerive{G'}^{*}$};
    \end{scope}

    \begin{scope}[shift={(11,0)}]
        \location{v1}{}{}
        \nonterminal{S1}{right}{v1}{S'}
        \location{v2}{right of = S1}{}

        \tentacle{S1}{1}{v1}{}
        \tentacle{S1}{2}{v2}{}
    \end{scope}

}
}

\end{center}
By Theorem~\ref{thm:ig-props}, this implies $H \in \iL{G'}(e_{S'})$.
However, this contradicts our assumption $\iL{G'}(e_{S'}) = \iL{G}(e_{S})$, because $H \notin \iL{G}(e_{S})$.
\qed
\end{proof}

\subsection{Decidability of Language Inclusion}\label{app:sec-backwardconfluent-inclusion}

We first formalize a crucial property informally stated in Section~\ref{sec:backwardconfluent}.

\begin{theorem}\label{thm:backwardconfluent-intersection}
  Let $\iG$ be a backward confluent IG. 
  Moreover, let $\iH,\iK \in \IHC$ such that 
  $\iH \nrDerive{\iG}$ and $\iK \nrDerive{\iG}$.
  Then
  \begin{align*}
          \Lang{\iG}{\iH} \cap \Lang{\iG}{\iK} \neq \emptyset ~\text{iff}~ \iH ~\text{is isomorphic to}~ \iK.
  \end{align*}
\end{theorem}

\begin{proof}
We first show that 
$\iL{\iG}(\iH) \cap \iL{\iG}(\iK) \neq \emptyset$ implies that $\iH$ and $\iK$ are isomorphic: 
\begin{align*}
& 
\iL{\iG}(\iH) \cap \iL{\iG}(\iK) \neq \emptyset
\\
~\Leftrightarrow~ &
\exists \iR : \E_{\iR}^{T} = \E_{\iR} ~\text{and}~ \iH \iDerive{\iG}^{*} \iR ~\text{and}~ \iK \iDerive{\iG}^{*} \iR 
\tag{Def.~\ref{def:language}} \\
~\Leftrightarrow~ &
\exists \iR : \E_{\iR}^{T} = \E_{\iR} ~\text{and}~ \iR \rDerive{\iG}^{*} \iH ~\text{and}~ \iR \rDerive{\iG}^{*} \iK 
\tag{Def. $\rDerive{\iG}$} \\
~\Leftrightarrow~ &
\exists \iR : \E_{\iR}^{T} = \E_{\iR} ~\text{and}~ \iH \in \rL{\iG}(\iR) ~\text{and}~ \iK \in \rL{\iG}(\iR)
\tag{$\iH,\iK \nrDerive{\iG}$} \\
~\Leftrightarrow~ &
\exists \iR : \E_{\iR}^{T} = \E_{\iR} ~\text{and}~ \iH \in \rL{\iG}(\iR) ~\text{and}~ \iK \in \rL{\iG}(\iR) ~\text{and}~ |\rL{\iG}(\iR)| = 1
\tag{Def.~\ref{def:backwardconfluent}} \\
~\Rightarrow~ &
\iH ~\text{isomorphic to}~ \iK.
\end{align*}
The converse direction, i.e. if $\iH$ is isomorphic to $\iK$ then $\iL{\iG}(\iH) \cap \iL{\iG}(\iK) \neq \emptyset$, is trivial.
\qed
\end{proof}

\noindent\textbf{Theorem~\ref{thm:backwardconfluent-inclusion}.} 
Let $\iG$ be a backward confluent IG.
Moreover, let $\iH, \iK \in \IHC$ such that $\iK \nrDerive{\iG}$. 
Then it is decidable whether 
$\Lang{\iG}{\iH} \subseteq \Lang{\iG}{\iK}$ holds.

\begin{proof}

Note that it is decidable whether $\iK \iDerive{\iG}^{*} \iH$ holds, because
$\iK \iDerive{\iG}^{*} \iH$ holds if and only if there exists a natural number
$0 \leq m \leq |\V_{\iH}| + |\E_{\iH}|$ such that $\iK \iDerive{\iG}^{m} \iH$ holds, where $\iDerive{\iG}^{m}$ means that we apply up to $m$ IG derivations.
We distinguish two cases:

\begin{enumerate}
\item Assume $\iK \iDerive{\iG}^{*} \iH$ holds.
By Theorem~\ref{thm:ig-props}(1) this implies $\iL{\iG}(\iH) \subseteq \iL{\iG}(\iK)$.

\item
Assume $\iK \iDerive{\iG}^{*} \iH$ does not hold.
Since $\iG$ is backward confluent there exists a unique $\iR \in \rL{\iG}(\iH)$.
Moreover, by our assumption $\iK \nrDerive{\iG}$, we have $\rL{\iG}(\iK) = \{ \iK \}$.
Then $\iR$ and $\iK$ are \emph{not} isomorphic. 
Hence,
\begin{align*}
& \iR ~\text{not isomorphic to}~ \iK \\ 
~\Leftrightarrow~ & 
\iL{\iG}(\iR) \cap \iL{\iG}(\iK) = \emptyset 
\tag{Lemma~\ref{thm:backwardconfluent-intersection}} \\
~\Rightarrow~ & 
\iL{\iG}(\iH) \cap \iL{\iG}(\iK) = \emptyset
\tag{$\iL{\iG}(\iH) \subseteq \iL{\iG}(\iR)$ by Theorem~\ref{thm:ig-props}(1)}
\end{align*}
\end{enumerate}
Thus $\iL{\iG}(\iH) \subseteq \iL{\iG}(\iK)$ holds if and only if $\iL{\iG}(\iH) = \emptyset$.
By Theorem~\ref{thm:ig-props}(3), it is decidable whether $\iL{\iG}(\iH) = \emptyset$.
\qed
\end{proof}

\subsection{Properties of Indexed Graph Grammars with Index Abstraction}\label{app:sec-cfg-props}

\subsubsection{Well-Formedness of Indices}
\label{app:wellformed}

\begin{lemma}\label{thm:wellformed}
If $\iH$ is well-formed and $\iH \iCon{\iC}^{*} \iK$ or $\iH \iDerive{\iG}^{*} \iK$ then $\iK$ is well-formed.
\end{lemma}

\begin{proof}
Let $\iC$ be a CFG and $\iG$ be an IG satisfying the assumptions at the beginning of Section~\ref{sec:global}.
Moreover, let $\iH$ be a well-formed IHC.

We first prove that $\iH \iCon{\iC} \iK$ implies that $\iK$ is well-formed:
\begin{align*}
& \iH \iCon{\iC} \iK \\
~\Leftrightarrow~ & 
\exists (\Rule{X}{\sigma}) \in \iC : \Ind_{\iH}(\E_{\iH}^{N}) \subseteq \I_{T}^{*} \I_{N} ~\text{and}~ \iK \Iso \Replace{\iH}{X}{\sigma} 
\tag{Def.~\ref{def:indexabstraction}} \\
~\Rightarrow~ & 
\exists (\Rule{X}{\sigma}) \in \iC : \Ind_{\iK}(\E_{\iK}^{N}) \subseteq \I_{T}^{*} \sigma \,+\, \I_{T}^{*} \I_{N} \\
~\Rightarrow~ & 
\Ind_{\iK}(\E_{\iK}^{N}) \subseteq (\I_{T} \setminus \{z\})^{*} (\I_{N} \cup \{z\})
\tag{$\sigma$ is well-formed} \\
~\Rightarrow~ & \iK ~\text{is well-formed}.
\end{align*}
Next we prove that $\iH \iDerive{\iG} \iK$ implies that $\iK$ is well-formed. 
By Definition~\ref{def:derive}, two cases arise:
\begin{enumerate}
\item 
There exists a rule $(\Rule{X,\sigma}{\iR}) \in \iG$ and an edge $e \in \E_{\iH}^{N}$ such that $\Ind_{\iH}(e) = \sigma$ and $\iK$ is isomorphic to $\Replace{\iH}{e}{\iR}$.
By assumption, $\Ind_{\iR}(e') \subseteq \I_T^{*} z$. 
Thus, for each $e' \in \E_{\iR}$, $\Ind_{\iR}(e)$ is well-formed.
Then, by Definition~\ref{def:ihr} and the fact that $\iH$ is well-formed, $\Ind_{\iK}(e'')$ is well-formed for each $e'' \in \E_{\iK} = (\E_{\iH} \setminus \{e\}) \cup \E_{\iK}$. 
Hence, $\iK$ is well-formed.

\item
There exists a rule $(\Rule{X,\sigma\iVar}{\iR}) \in \iG$, an edge $e \in \E_{\iH}^{N}$, and a sequence $\rho \in \I^{+}$ 
such that $\Ind_{\iH}(e) = \sigma\rho$ and $\iK$ is isomorphic to $\Replace{\iH}{e}{\Replace{\iR}{\iVar}{\rho}}$.
Since $\iH$ is well-formed, so is $\rho$.
By assumption $\Ind_{\iR}(\E_{\iR}^{N}) \subseteq \I_{T}^{*}\iVar$.
This means that $\Replace{\iR}{\iVar}{\rho}$ is well-formed.
Hence, $\Replace{\iH}{e}{\Replace{\iR}{\iVar}{\rho}}$ and thus $\iK$ is well-formed.
\end{enumerate}
It then follows by a straightforward induction on the length of derivations $\iDerive{\iG}^{*}$ and global derivations $\iCon{\iC}^{*}$ that 
$\iH \iCon{\iG}^{*} \iK$ implies that $\iK$ is well-formed and
$\iH \iDerive{\iG}^{*} \iK$ implies that $\iK$ is well-formed.
\qed
\end{proof}

\subsubsection{Analogon to Theorem~\ref{thm:ig-props} for Global Derivations}
\label{app:cfg-props}

\begin{theorem}\label{thm:cfg-props}
  Let $\iC$ be a CFG and $\iH,\iK \in \IHC$.
  Then:
  \begin{enumerate}
    \item $\iH \iCon{\iC}^{*} \iK$ implies $\iGL{\iC}(\iK) \subseteq \iGL{\iC}(\iH)$.
    \item $\iGL{\iC}(\iH) = 
            \begin{cases}
                    \{ \iH \} &~\text{if}~ \Ind_{\iH}(\E_{\iH}^{N}) \subseteq \I_{T}^{*} \iBot \\ 
                \bigcup_{\iH \iCon{\iC} \iK} \iGL{\iC}(\iK) &~\text{otherwise.}
            \end{cases}$
    \item If $\iC$ contains a single nonterminal, i.e. $|\I_{N}|=1$, it is decidable whether $\iGL{\iC}(\iH) = \emptyset$ holds.
    \item 
          If $\iC$ contains no rule of the form $\Rule{X}{Y}$, where $X,Y \in \I_{N}$, then the global inverse language $\rGL{\iC}(\iH)$ is non-empty and finite.
  \end{enumerate}
\end{theorem}

We write $\iH \iCon{\iC}^{n} \iK$ to denote that $\iC$ globally derives $\iK$ from $\iH$ in exactly $n \in \Nats$ steps.
Formally, 
\begin{itemize}
\item $\iH \iCon{\iC}^{0} \iH$, and
\item $\iH \iCon{\iC}^{n+1} \iK$ iff $\exists R : \iH \iCon{\iC} \iR$ and $\iR \iCon{\iC} \iK$.
\end{itemize}
Then $\iCon{\iC}^{*} = \bigcup_{n \in \Nats} \iCon{\iC}^{n}$.

\bigskip
\noindent\textbf{Theorem~\ref{thm:cfg-props}(1).}
Let $\iC$ be a CFG and $\iH,\iK \in \IHC$. Then
$\iH \iCon{\iC}^{*} \iK$ implies $\iGL{\iC}(\iK) \subseteq \iGL{\iC}(\iH)$.

\begin{proof}
Similarly to the proof of Theorem~\ref{thm:ig-props}(1), we show for all natural numbers $n \in \Nats$ that $\iH \iCon{\iC}^{n} \iK$ implies $\iGL{\iC}(\iK) \subseteq \iGL{\iC}(\iH)$.
By induction on $n$.
\noindent\textbf{I.B.} For $n = 0$, we have
\begin{align*}
& \iH \iCon{\iC}^{0} \iK \\
~\Leftrightarrow~ &
\iH \Iso \iK
\tag{Def.~$\iCon{\iC}^{0}$} \\
~\Rightarrow~ &
\iGL{\iC}(\iH) = \iGL{\iC}(\iK).
\end{align*}

\noindent\textbf{I.H.} Assume for an arbitrary, but fixed $n \in \Nats$ that $\iH \iCon{\iC}^{n} \iK$ implies $\iGL{\iC}(\iH) \subseteq \iGL{\iC}(\iK)$.

\noindent\textbf{I.S.} For $n \mapsto n+1$ we have
\begin{align*}
& 
\iH \iCon{\iC}^{n+1} \iK 
\\
~\Leftrightarrow~ &
\exists \iR : \iH \iCon{\iC} \iR ~\text{and}~ \iR ~\iCon{\iC}^{n} \iK 
\tag{Def.~$\iCon{\iC}^{n+1}$} \\
~\Rightarrow~ &
\exists \iR : \iH \iCon{\iC} \iR ~\text{and}~ \iGL{\iC}(\iK) \subseteq \iGL{\iC}(\iR).
\tag{I.H.}
\end{align*}
It then suffices to show that $\iGL{\iC}(\iR) \subseteq \iGL{\iC}(\iH)$ holds as
\[ \iGL{\iC}(\iK) \subseteq \iGL{\iC}(\iR) ~\text{and}~ \iGL{\iC}(\iR) \subseteq \iGL{\iC}(\iH) ~\text{implies}~ \iGL{\iC}(\iK) \subseteq \iGL{\iC}(\iH). \]
The remaining proof obligation is shown as follows:
\begin{align*}
& \iGL{\iC}(\iR) \\
~=~ &
\{ F ~|~ \iR \iCon{\iC}^{*} F ~\text{and}~ \Ind_{F}(\E_{F}) \subseteq \I_{T}^{+} \} 
\tag{Def.~\ref{def:globallanguage}} \\
~=~ &
\{ F ~|~ \iH \iCon{\iC} \iR \iCon{\iC}^{*} F ~\text{and}~ \Ind_{F}(\E_{F}) \subseteq \I_{T}^{+} \} 
\tag{$\iH \iCon{\iC} \iR$ possible by assumption} \\
~\subseteq~ &
\{ F ~|~ \iH \iCon{\iC}^{*} F ~\text{and}~ \Ind_{F}(\E_{F}) \subseteq \I_{T}^{+} \}
\\
~=~ &
\iGL{\iC}(\iH).
\tag{Def.~\ref{def:globallanguage}}
\end{align*}
\qed
\end{proof}

\bigskip
\noindent\textbf{Theorem~\ref{thm:cfg-props}(2).}
Let $\iC$ be a CFG and $\iH,\iK \in \IHC$. Then
\[\iGL{\iC}(\iH) = 
\begin{cases}
        \{ \iH \} &~\text{if}~ \Ind_{\iH}(E_{\iH}^{N}) \subseteq \I_T^{*} \iBot \\
    \bigcup_{\iH \iCon{\iC} \iK} \iGL{\iC}(\iK) &~\text{otherwise.}
\end{cases}\]

\begin{proof}
\noindent\textbf{Case 1:}
Assume that $\Ind_{\iH}(E_{\iH}^{N}) \subseteq \I_{T}^{*} \iBot$.
By Definition~\ref{def:indexabstraction} this means that $\iH \niCon{\iC}$. Then
\begin{align*}
& \iGL{\iC}(\iH) \\
~=~ &
\{ \iK ~|~ \iH \iCon{\iC}^{*} \iK ~\text{and}~ \Ind_{\iK}(\E_{\iK}) \subseteq \I_{T}^{+} \}
\tag{Def.~\ref{def:globallanguage}} \\
~=~ &
\{ \iK ~|~ \iK \iCon{\iC}^{0} \iK \} 
\tag{$\iH \niCon{\iC}$} \\
~=~ &
\{ \iH \}.
\tag{Def.~$\iCon{\iC}^{0}$, $\I_{N},\I_{T}$ disjoint}
\end{align*}

\noindent\textbf{Case 2:}
Assume that $\Ind_{\iH}(E_{\iH}^{N}) \not\subseteq \I_{T}^{*} \iBot$.
Then there exists $X \in \I_{N}$ and $e \in E_{\iH}$ such that $\Ind_{\iH}(e) \subseteq \I_{T}^{*} X$.
By Definition~\ref{def:globallanguage}, $\iH \notin \iGL{\iC}(\iH)$.
If $\iGL{\iC}(\iH) = \emptyset$, there is nothing to show. 
Thus assume $\iGL{\iC}(\iH) \neq \emptyset$.
By Definition~\ref{def:indexabstraction} this means that there exists a $\iK$ such that $\iH \iCon{\iC} \iK$. 
Then
\begin{align*}
& \iGL{\iC}(\iH) \\
~=~ & 
\{ \iR ~|~ \iH \iCon{\iC}^{*} \iR ~\text{and}~ \Ind_{\iR}(e) \subseteq \I_{T}^{+} \} 
\tag{Def.~\ref{def:globallanguage}} \\
~=~ & 
\{ \iR ~|~ \exists n \in \Nats : \iH \iCon{\iC}^{n+1} \iR ~\text{and}~ \Ind_{\iR}(e) \subseteq \I_{T}^{+} \} 
\tag{$\iH \in \iGL{\iC}(\iH)$ and $\iH \iCon{\iC}^{0} \iR$ iff $\iH \Iso \iR$} \\
~=~ & 
\{ \iR ~|~ \exists \iK \exists n \in \Nats : \iH \iCon{\iC} \iK \iCon{\iC}^{n} \iR ~\text{and}~ \Ind_{\iR}(e) \subseteq \I_{T}^{+} \} 
\tag{Def. $\iCon{\iC}^{n+1}$} \\
~=~ &
\bigcup_{\iH \iCon{\iC} \iK} \{ \iR ~|~ \exists n \in \Nats : \iK \iCon{\iC}^{n} \iR ~\text{and}~ \Ind_{\iR}(e) \subseteq \I_{T}^{+} \} 
\\ 
~=~ &
\bigcup_{\iH \iCon{\iC} \iK} \{ \iR ~|~ \iK \iCon{\iC}^{*} \iR ~\text{and}~ \Ind_{\iR}(e) \subseteq \I_{T}^{+} \} 
\tag{Def. $\iCon{\iC}^{*}$} \\
~=~ &
\bigcup_{\iH \iCon{\iC} \iK} \iGL{\iC}(\iK).
\tag{Def.~\ref{def:globallanguage}} 
\end{align*}
\qed
\end{proof}

\bigskip
\noindent\textbf{Theorem~\ref{thm:cfg-props}(3).}
Let $\iC$ be a CFG with a single nonterminal, i.e. $|\Ind_{\iH}|=1$, and $\iH \in \IHC$. Then
it is decidable whether $\iGL{\iC}(\iH) = \emptyset$ holds.

\begin{proof}
Let $\iC$ be a CFG and $\iH \in \IHC$.
By assumption (cf. Section~\ref{sec:global}), $\iH$ is well-formed.
That is, $\Ind_{\iH}(\E_{\iH}) \subseteq (\I_{T} \setminus \{z\})^{*} (\I_{N} \cup \{z\})$.
Furthermore, again by assumption, we know that for each rule $(\Rule{X}{\sigma}) \in \iC$, $\sigma$ is well-formed.
We distinguish three cases:
\begin{enumerate}
        \item $\Ind_{\iH}(\E_{\iH}) \subseteq (\I_{T} \setminus \{z\})^{*} \{ z \} \subseteq \I_{T}^{+}$. 
                By Definition~\ref{def:globallanguage}, this means that $\iH \in \iGL{\iC}(\iH)$ and thus $\iGL{\iC}(\iH) \neq \emptyset$.
        \item For the single nonterminal $X$ of $\iC$, we have $\Ind_{\iH}(\E_{\iH}) \subseteq (\I_{T} \setminus \{z\})^{*} \{ X \}$.
              Then
\begin{align*}
& \iGL{\iC}(\iH) \neq \emptyset \\
~\Leftrightarrow~ & \exists \iK : \iH \iCon{\iC}^{*} \iK 
~\text{and}~ \Ind_{\iK}(\E_{\iK}) \subseteq \I_{\T}^{+} 
\tag{Def.~\ref{def:globallanguage}} \\
~\Leftrightarrow~ & \exists \rho \in \iL{\iC}(X) \exists \iK : \iH \iCon{\iC}^{*} \iK ~\text{and}~ \iK ~\text{isomorphic to}~ \Replace{\iH}{X}{\rho} 
\tag{Def.~\ref{def:indexabstraction}, assumption for case 2, context-freeness of $\iC$} \\
~\Leftrightarrow~ & \iL{\iC}(X) \neq \emptyset~,
\end{align*}
where the last line is an instance of the emptiness problem for context-free string grammars.
This problem is well-known to be decidable in linear time.\footnote{For details confer Rozenberg, Salomaa: \enquote{Handbook of Formal Languages}, Vol. 1, 1997} 
\item Neither the first, nor the second case holds. By Definition~\ref{def:globallanguage} this means that $\iH \notin \iGL{\iC}(\iH)$. Moreover, by Definition~\ref{def:indexabstraction}, there exists no $\iK$ such that
      $\iH \iCon{\iC} \iK$. Hence, $\iGL{\iC}(\iH) = \emptyset$.
\end{enumerate}
Hence, we can decide whether $\iGL{\iC}(\iH) = \emptyset$ holds by first checking which of the above cases applies to the finite set $\Ind_{\iH}(\E_{\iH})$ and then proceed as shown for each case.
\qed
\end{proof}

\bigskip
\noindent\textbf{Theorem~\ref{thm:cfg-props}(4).}
Let $\iC$ be a CFG containing no rule of the form $\Rule{X}{Y}$ for $X,Y \in \I_{N}$.
Moreover, let $\iH,\iK \in \IHC$. Then the global inverse language $\rGL{\iC}(\iH)$ is non-empty and finite.
\begin{proof}
Recall that $\iH \iAbs{\iC} \iK$ holds iff $\iK \iCon{\iC} \iH$ holds.
Let 
\[ \textit{length}(\iH) = \max \{ |\sigma| ~|~ \sigma \in \Ind_{\iH}(\E_{\iH}) \} \]
be the maximal length of all indices in $\iH$.
By Definition~\ref{def:indexabstraction}, we know that $\iK \iCon{\iC} \iH$ implies
\begin{align*}
  0 \leq \textit{length}(\iH) \leq \textit{length}(\iK).
\end{align*}
By the premise of the theorem and the fact that all indices are well-formed, we have:
\begin{align*}
& \textit{length}(\iH) = \textit{length}(\iK)
\\
~\Rightarrow~ &
\Ind_{\iK}(\E_{\iK}) \subseteq \I_{T} z 
\tag{$\iH,\iK$ well-formed, premise} \\
~\Rightarrow~ &
\iK \niCon{\iC}.
\tag{Def.~\ref{def:indexabstraction}}
\end{align*}
In particular, this means that 
\begin{align*}
  \exists \iR : \iK \iCon{\iC} \iR ~\text{implies}~ 0 \leq \textit{length}(\iH) < \textit{length}(\iK). \tag{$\spadesuit$}
\end{align*}
We now prove that $\rGL{\iC}(\iH)$ is finite:
\begin{align*}
& \rGL{\iC}(\iH) \\
~=~ & 
\{ \iK ~|~ \iH \iAbs{\iC}^{*} \iK ~\text{and}~ \iK \niAbs{\iC} \} 
\tag{Def.~\ref{def:globallanguage}} \\
~=~ & 
\{ \iK ~|~ \iK \iCon{\iC}^{*} \iH ~\text{and}~ \iK \niAbs{\iC}\} 
\tag{Def.~\ref{def:indexabstraction}} \\
~=~ &
\{ \iK ~|~ \exists n \in \Nats : \iK \iCon{\iC}^{n} \iH ~\text{and}~ \iK \niAbs{\iC}\} 
\tag{Def.~$\iCon{\iC}^{*}$} \\
~=~ &
\{ \iK ~|~ \exists n \leq m : \iK \iCon{\iC}^{n} \iH ~\text{and}~ \iK \niAbs{\iC}\} 
\tag{set $m = 1 + \textit{length}(\iH)$ and apply $(\spadesuit)$}~,
\end{align*}
where the last set is finite.
It remains to prove that $\rGL{\iC}(\iH) \neq \emptyset$:
First, note that $\iH \iCon{\iC}^{0} \iH$ always holds.
Further, by $(\spadesuit)$, after at most $m$ inverse global derivation steps, i.e. $\bigcup_{0 \leq n \leq m} \iAbs{\iC}^{n}$, no further inverse global derivation is possible.
Hence, $\rGL{\iC}(\iH) \neq \emptyset$.
\qed
\end{proof}

\subsubsection{Commutability of IG Derivations and Gobal Derivations}
\label{app:indexcommutable}

\begin{theorem}\label{thm:indexcommutable}
        \[ \iH (\iDerive{\iG} \FComp \iCon{\iC}) \iK ~\text{implies}~ \iH (\iCon{\iC} \FComp \iDerive{\iG}) \iK.\]
\end{theorem}
\begin{proof}

\begin{align*}
& \iH (\iDerive{\iG} \FComp \iCon{\iC}) \iK \\
~\Rightarrow~ & 
\exists \iR : \iH \iDerive{\iG} \iR ~\text{and}~ \iR \iCon{\iC} \iK 
\tag{Def.~$\FComp$} \\
~\Rightarrow~ & 
\exists (\Rule{X}{\sigma}) \in \iC \exists \iR \,:\,
\iH \iDerive{\iG} \iR 
\tag{Def.~\ref{def:indexabstraction}} \\
& \qquad \text{and}~ \Ind_{\iR}(\E_{\iR}^{N}) \subseteq \I_{T}^{*}\I_{N} ~\text{and}~ \iK \Iso \Replace{\iR}{X}{\sigma}
\\
~\Rightarrow~ & 
\exists (\Rule{X}{\sigma}) \in \iC \exists \iR \,:\,
\iH \iDerive{\iG} \iR ~\text{and}~ \Ind_{\iH}(\E_{\iH}^{N}) \subseteq \I_{T}^{*} X 
\tag{Ass. Sec.~\ref{sec:global}} \\
& \qquad \text{and}~ \Ind_{\iR}(\E_{\iR}^{N}) \subseteq \I_{T}^{*}\I_{N} ~\text{and}~ \iK \Iso \Replace{\iR}{X}{\sigma}
\\
~\Rightarrow~ & 
\exists (\Rule{X}{\sigma}) \in \iC \, \exists \iR \, \exists (\Rule{Y,i\iVar}{\iR'}) \in \iG \, \exists e' \in \E_{\iH}^{Y} \, \exists j \in \I_{\T}^{*}X \,:\, 
\tag{Def.~\ref{def:derive}, case 2} \\
& \qquad \Ind_{\iH}(e') = ij ~\text{and}~ \iR \Iso \Replace{\iH}{e'}{\Replace{\iR'}{\iVar}{j}} ~\text{and}~ \Ind_{\iH}(\E_{\iH}^{N}) \subseteq \I_{T}^{*} \I_{N} 
\\
& \qquad \text{and}~ \Ind_{\iR}(\E_{\iR}^{N}) \subseteq \I_{T}^{*}\I_{N} ~\text{and}~ \iK \Iso \Replace{\iR}{X}{\sigma}
\\
~\Rightarrow~ & 
\exists (\Rule{X}{\sigma}) \in \iC \, \exists (\Rule{Y,i\iVar}{\iR'}) \in \iG \, \exists e' \in \E_{\iH}^{Y} \, \exists j \in \I_{\T}^{*}\I_{N} \,:\, 
\tag{remove $\iR$} \\
& \qquad \iK \Iso \Replace{(\Replace{\iH}{e'}{\Replace{\iR'}{\iVar}{j}})}{X}{\sigma} \\ 
& ~\text{and}~ \Ind_{\iH}(\E_{\iH}^{N}) \subseteq \I_{T}^{*} \I_{N} ~\text{and}~ \Ind_{\iH}(e') = ij 
\\
~\Rightarrow~ & 
\exists (\Rule{X}{\sigma}) \in \iC \, \exists (\Rule{Y,i\iVar}{\iR'}) \in \iG \, \exists e' \in \E_{\iH}^{Y} \, \exists j \in \I_{\T}^{*}X \,:\, 
\\
& \qquad \iK \Iso \Replace{(\Replace{\iH}{X}{\sigma})}{e'}{\Replace{\iR'}{\iVar}{\ell}} \tag{set $\ell = \Replace{j}{X}{\sigma}$} \\ 
& ~\text{and}~ \Ind_{\iH}(\E_{\iH}^{N}) \subseteq \I_{T}^{*} \I_N ~\text{and}~ \Ind_{\iH}(e') = ij
\\
~\Rightarrow~ & 
\exists \iR \, \exists (\Rule{X}{\sigma}) \in \iC \, \exists (\Rule{Y,i\iVar}{\iR'}) \in \iG \, \exists e' \in \E_{\iR}^{Y} \, \exists \ell \in \I_{\T}^{*}X \,:\, 
\\
& \qquad \iR \Iso \Replace{\iH}{X}{\sigma} ~\text{and}~ \Ind_{\iH}(\E_{\iH}^{N}) \subseteq \I_{T}^{*} \I_{N} \\
& \qquad ~\text{and}~ \iK \Iso \Replace{\iR}{e'}{\Replace{\iR'}{\iVar}{\ell}} ~\text{and}~ \Ind_{\iR}(e') = i\ell
\\
~\Rightarrow~ & 
\exists \iR \, \exists (\Rule{Y,i\iVar}{\iR'}) \in \iG \, \exists e' \in \E_{\iR}^{Y} \, \exists \ell \in \I_{\T}^{*}X \,:\, 
\tag{Def.~\ref{def:indexabstraction}}  \\
& \qquad \iH \iCon{\iC} \iR ~\text{and}~ \iK \Iso \Replace{\iR}{e'}{\Replace{\iR'}{\iVar}{\ell}} ~\text{and}~ \Ind_{\iR}(e') = i\ell
\\
~\Rightarrow~ & 
\exists \iR \,:\, \iH \iCon{\iC} \iR ~\text{and}~ \iR \iDerive{\iG} \iK
\tag{Def.~\ref{def:derive}} \\
~\Rightarrow~ &
\iH (\iCon{\iC} \FComp \iDerive{\iG}) \iK
\tag{Def.~$\FComp$}
\end{align*}
\qed
\end{proof}

\subsection{Analysis with Index Abstraction}\label{app:sec-cfg-analysis}

We now refine our analysis presented in Section~\ref{sec:analysis} such that only relevant relationships between indices are kept.
Our new abstraction places the inverse global language of a CFG on top of our original abstraction guided by an IG.
Hence, given an IG $\iG$ and a CFG $\iC$, we change the original setting (cf. Figure~\ref{fig:domains}, p.~\SUMPAGE) as follows: 
\begin{itemize}
        \item Concrete domain:  $(\DomCon = \PS{\underbrace{\Lang{\iG}{\iGL{\iC}(\IHC)}}_{\text{concrete IHCs}}}, \subseteq)$
        \item Abstract domain:  $(\DomAbs = \PS{\underbrace{\rGL{\iC}({\iLang{\iG}{\IHC}})}_{\text{fully abstract IHCs}}}, \LOrder{})$
  \item Concretization and abstraction: $\con{\iG} ~=~ \iGL{\iC} \FComp \iL{\iG}$, $\abs{\iG} ~=~ \rL{\iG} \FComp \rGL{\iC}$
  \item Ordering of abstract domain: $\iH \LOrder{} \iK$ iff $\con{\iG}(\iH) \subseteq \con{\iG}(\iK)$
\end{itemize}

Materialization and canonicalization now formally also depend on the user-provided CFG $\iC$.
We experienced that aggressively abstracting indices 
and partially concretizing them whenever needed works well in practice.

We remark that all results from Sections~\ref{sec:analysis} and~\ref{sec:backwardconfluent} can be lifted to our refined analysis.
In particular, the analysis is sound and --- in the case of backward confluent grammars --- the inclusion problem is decidable:

\begin{theorem}[Soundness]\label{thm:indexed-sound}
The refined analysis from above is sound. That is,
$\forall \pP \in \Progs \,:\, \con{\iG} \FComp \PSem{P} ~\FOrder~ \PASem{P} \FComp \con{\iG}$.
\end{theorem}

\begin{theorem}
Let $\iC$ and $\iG$ be a backward confluent CFG and IG, respectively.\footnote{$\iC$ is backward confluent if for all well-formed IHCs $\iH$, we have $|\rGL{\iC}(\iH)| = 1$.}
Moreover, let $\iH,\iK$ be well-formed IHCs such that $\iK \niAbs{\iC}$ and $\iK \nrDerive{\iG}$.
Then it is decidable whether $\iH \LOrder{} \iK$ holds.
\end{theorem}

Both theorems are proven analogously to Theorem~\ref{thm:sound} and Corollary~\ref{thm:backwardconfluent-inclusion}. 

\begin{theorem}
  If $\iC$ and $\iG$ are backward confluent then 
\begin{align*}
(\PS{\iGL{\iC}(\IHC)}, \subseteq ) 
\galois{\rGL{\iC}}{\iGL{\iC}}
(\PS{\rGL{\iC}(\IHC)}, \LOrder{\iC} )
\end{align*}
is a Galois connection.
\end{theorem}

\subsection{Grammar for Doubly-Linked List Segments}
\label{app:grammars}

Please find below two additional indexed graph grammars.
The first specifies doubly-linked lists.
The second provides additional rules for balanced trees that have been omitted in Figure~\ref{fig:hrg_balanced_trees} (although they do not affect the overall language).

\bigskip
\scalebox{0.8}{
\tikzdefault{
   
    \node (text) {$\texttt{DLL} ~\rightarrow~$};
    \node[below of = text, node distance = 15mm] (text2) {$\texttt{DLL} ~\rightarrow~$};
    \node[below of = text2, node distance = 15mm] (text3) {$\texttt{DLL} ~\rightarrow~$};
    \node[below of = text3, node distance = 18mm] (text4) {$\texttt{DLL} ~\rightarrow~$};

    \begin{scope}
    \location{v1}{right of = text}{2}
    \location{v2}{right of = v1}{3}
    \location{v3}{right of = v2}{4}
    \location{v4}{right of = v3}{5}
    \nilext{nil}{right of = v4}
    
    \lpointer{v2}{bend left}{e2}{v1}{}{below}
    \lpointer{v2}{bend left}{e3}{v3}{}{above}
    \lpointer{v3}{bend left}{e4}{v2}{}{below}
    \lpointer{v3}{bend left}{e5}{v4}{}{above}
    \end{scope}
    
    \begin{scope}
    \location{v1}{right of = text2}{2}
    \location{v2}{right of = v1}{3}
    \location{v3}{right of = v2}{}
    \nonterminal{L}{right}{v3}{\texttt{DLL}}
    \location{v4}{right of = L}{4}
    \location{v5}{right of = v4}{5}
    \nilext{nil}{right of = v5}

    \lpointer{v2}{bend left}{e2}{v1}{}{below}
    \lpointer{v2}{bend left}{e3}{v3}{}{above}
    

    \tentacle{L}{1}{nil}{bend left=30}
    \tentacleB{L}{2}{v2}{bend left}
    \tentacle{L}{3}{v3}{}
    \tentacle{L}{4}{v4}{}
    \tentacleB{L}{5}{v5}{bend right}
    \end{scope}

    \begin{scope}
    \location{v1}{right of = text3}{2}
    \location{v2}{right of = v1}{3}
    \nonterminal{L}{right}{v2}{\texttt{DLL}}

    \location{v3}{right of = L}{}
    \location{v4}{right of = v3}{4}
    \location{v5}{right of = v4}{5}
    \nilext{nil}{right of = v5}

    \lpointer{v4}{bend left}{e2}{v3}{}{below}
    \lpointer{v4}{bend left}{e3}{v5}{}{above}
    

    \tentacleB{L}{1}{nil}{bend right=25}
    \tentacleB{L}{2}{v1}{bend left}
    \tentacle{L}{3}{v2}{}
    \tentacle{L}{4}{v3}{}
    \tentacle{L}{5}{v4}{bend left}
    \end{scope}

    \begin{scope}
    \location{v1}{right of = text4}{2}
    \location{v2}{right of = v1}{3}
    \nonterminal{A}{right}{v2}{\texttt{DLL}}
    \location{v3}{right of = A}{}
    \location{v4}{right of = v3}{}
    \nonterminal{B}{right}{v4}{\texttt{DLL}}
    \location{v5}{right of = B}{4}
    \location{v6}{below of = v5, node distance=10mm}{5}
    \nilext{nil}{below of = v3, node distance=10mm}

    \tentacleB{A}{1}{nil}{}
    \tentacleB{A}{2}{v1}{bend left}
    \tentacle{A}{3}{v2}{}
    \tentacleB{A}{4}{v3}{}
    \tentacle{A}{5}{v4}{bend left}

    \tentacleB{B}{1}{nil}{bend left}
    \tentacleB{B}{2}{v3}{bend left}
    \tentacle{B}{3}{v4}{}
    \tentacle{B}{4}{v5}{}
    \tentacleB{B}{5}{v6}{}
    \end{scope}

}
}

\bigskip
\tikzdefault{

  \begin{scope}[shift={(0,0)}]
    \node (rule) {$\Pt, \iVar \rightarrow$};

    \location{0}{below of = rule, node distance = 0.7cm}{};
    \location{1}{below left of = 0}{2};
    \location{2}{below right of = 0}{};
    \nonterminal{4}{below}{2}{\B,\iVar};
    \nilext{nil}{below of = 0, node distance = 2.5cm};

    \lpointer{0}{\pEdge,\pColor}{}{nil}{}{left};

    \lpointer{0}{\lEdge,\lColor, bend right, above}{}{1}{}{};
    \lpointer{1}{\pEdge,\pColor, bend right, below}{}{0}{}{};

    \lpointer{0}{\rEdge,\rColor, bend left, above}{}{2}{}{};
    \lpointer{2}{\pEdge,\pColor, bend left, below}{}{0}{}{};

    \tentacleR{4}{2}{2}{};
    \tentacleR{4}{1}{nil}{bend left};
  \end{scope}

  \begin{scope}[shift={(5.0,0)}]
    \node (rule) {$\Pt, \iVar \rightarrow$};

    \location{0}{below of = rule, node distance = 0.7cm}{};
    \location{1}{below left of = 0}{2};
    \location{2}{below right of = 0}{};
    \nonterminal{4}{below}{2}{\B,\iVar};
    \nilext{nil}{below of = 0, node distance = 2.5cm};
    
    \lpointer{0}{\lEdge,\lColor, bend right, above}{}{1}{}{};
    \lpointer{1}{\pEdge,\pColor, bend right, below}{}{0}{}{};

    \lpointer{0}{\rEdge,\rColor, bend left, above}{}{2}{}{};
    \lpointer{2}{\pEdge,\pColor, bend left, below}{}{0}{}{};

    \tentacleR{4}{2}{2}{};
    \tentacleR{4}{1}{nil}{bend left};

    \nonterminal{5}{node distance = 1.5cm, right}{rule}{\Pt, \Z\iVar};
    \tentacle{0}{2}{5}{}; 
    \tentacleR{nil}{1}{5}{bend right=60}; 
  \end{scope}

}

\end{document}